\documentclass[journal]{IEEEtran}

\usepackage[utf8]{inputenc} 
\usepackage[T1]{fontenc}
\usepackage{url}
\usepackage{ifthen}
\usepackage{cite}
\usepackage[cmex10]{amsmath} 
\usepackage{amsthm}
\usepackage{amssymb}
\usepackage{amsfonts}
\usepackage{mathtools}
\usepackage{graphicx}
\usepackage{algorithmic}
\usepackage{algorithm}
\usepackage{array}
\usepackage{bbm, dsfont}
\usepackage{textcomp}
\usepackage{subcaption}
\usepackage{stfloats}
\usepackage{verbatim}
\usepackage{enumerate}
\usepackage[inline]{enumitem}
\usepackage{graphicx}
\usepackage{tabularx, multirow}
\usepackage{tikz}
\usetikzlibrary{shapes,arrows,positioning}
\usepackage{microtype}
\usepackage{booktabs}

\usepackage{pgfplots}
\usepgfplotslibrary{colorbrewer}

\renewcommand{\b}{\boldsymbol{b}}

\newcommand{\x}{\boldsymbol{x}}
\newcommand{\X}{\boldsymbol{X}}
\newcommand{\y}{\boldsymbol{y}}
\newcommand{\Y}{\boldsymbol{Y}}
\renewcommand{\u}{\boldsymbol{u}}
\newcommand{\U}{\boldsymbol{U}}

\newcommand{\p}{\boldsymbol{p}}
\newcommand{\q}{\boldsymbol{q}}

\newcommand{\Z}{\mathbb{Z}}
\newcommand{\F}{\mathbb{F}}

\renewcommand{\O}{\mathcal{O}}
\DeclareMathOperator{\sign}{sgn}
\newcommand*{\rv}[1]{\mathsf{#1}}
\newcommand*{\Rv}[1]{\boldsymbol{\mathsf{#1}}}

\newtheorem{theorem}{Theorem}
\newtheorem{lemma}{Lemma}
\newtheorem{claim}{Claim}

\definecolor{customgreen}{RGB}{77, 155, 74}
\definecolor{custombrown}{RGB}{166, 86, 40}
\usetikzlibrary{shapes.geometric}
\pgfplotsset{compat=newest}
\pgfdeclareplotmark{mystar}{
    \node[star,star point ratio=2.25,minimum size=4pt,
          inner sep=0pt,draw=custombrown,solid,fill=none] {};
}
\begin{document}
\title{GC+ Code: A Systematic Short Blocklength Code for Correcting Random Edit Errors in DNA Storage} 


\author{%
  \IEEEauthorblockN{Serge Kas Hanna}
  \IEEEauthorblockA{
\\ C\^{o}te d'Azur University, CNRS, I3S, Sophia Antipolis, France \\
                    Email: serge.kas-hanna@\{univ-cotedazur.fr, cnrs.fr\} 
                    \thanks{This paper builds on our conference version presented at the 2024 IEEE International Symposium on Information Theory (ISIT)~\cite{isit2024}, with significant improvements to the code construction and additional analysis.}
}
\vspace{-0.5cm}
}

\maketitle

\begin{abstract}
Storing digital data in synthetic DNA faces challenges in ensuring data reliability in the presence of edit errors—deletions, insertions, and substitutions—that occur randomly during various stages of the storage process. Current limitations in DNA synthesis technology also impose the use of short DNA sequences, highlighting the particular need for short edit-correcting codes. Motivated by these factors, we introduce a systematic code designed to correct random edits while adhering to typical length constraints in DNA storage. We evaluate its performance both theoretically and through simulations, and assess its integration within a DNA storage framework, revealing promising results.
\end{abstract}

\section{Introduction}
DNA storage has emerged as a promising medium for next-generation storage systems due to its high density (\mbox{$10^{15}$-$10^{20}$} bytes per gram of DNA~\cite{DNA}) and long-term durability (thousands of years~\cite{grass2015robust}). One of the main challenges in DNA storage is ensuring data reliability in the presence of {\em edit} errors, i.e., deletions, insertions, and substitutions, which may occur during various stages of the storage process. A potential source of edit errors is sequencing and synthesis noise, with the actual error rate influenced by factors such as the technologies used and the length of the DNA sequence. Due to the limitations of current DNA synthesis technologies, binary data is typically encoded in the form of several short DNA sequences, called oligos, usually a few hundred nucleotides long, to minimize synthesis noise~\cite{heckel2019characterization}. In this setting, a common approach to enhance reliability involves using a combination of an inner and outer error-correction code. Ideally, the inner code should be a short code capable of correcting edit errors within the oligos, while the outer code can be a longer erasure/substitution code designed to recover lost oligos or correct residual errors from the inner code.

Designing error-correction codes for edit errors is a fundamental problem in coding theory, dating back to the 1960s~\cite{VT65, L66}. This problem has gained increased interest in recent years, with numerous works dedicated to constructing codes for correcting: only deletions {\em or} insertions, e.g.,~\cite{B16,SimaIT,SimaSYS,Sch17}, deletions {\em and} substitutions~\cite{gabrys2017codes,smagloy2023single,gabrys2022beyond,song2022systematic}, and sticky insertions/deletions~\cite{mahdavifar2017asymptotically,10206803}. However, much less is known about codes correcting all three types of edits simultaneously~\cite{cai2021correcting,tang2023correcting}. The aforementioned works focus on correcting adversarial (i.e., worst-case) errors with zero-error decoding, often relying on asymptotics that apply to scenarios with a small number of errors and large code lengths. While such assumptions are typical in coding theory, they do not naturally extend to DNA storage systems, which require different considerations. Specifically, in addition to the need for short codes due to synthesis limitations, edit errors in DNA storage are known to be of random nature and potentially of large quantity. Some earlier studies have proposed concatenated coding schemes for correcting random edit errors, e.g.,~\cite{davey2001reliable,ratzer2005marker}; however, the code lengths in these constructions are also large, typically in the order of several thousands.

Due to the lack of suitable edit-correcting codes, standard approaches in the literature rely on sequencing redundancy to correct edit errors in DNA storage. High-throughput sequencing and amplification can generate many reads per oligo, introducing sequencing redundancy analogous to repetition coding. This redundancy can be exploited in various ways, such as using sequence alignment algorithms to correct edits via majority voting. While leveraging sequencing redundancy has been shown to enhance reliability in several studies~\cite{goldman2013towards, grass2015robust, blawat2016forward, erlich2017dna, yazdi2017portable, chandak2020overcoming, press2020hedges, trellisBMA, maarouf2022concatenated, welzel2023dna, hamoum2023synchronization}, generating large numbers of redundant reads also has drawbacks and limitations. Namely, it is a resource-intensive process that incurs high read costs, both in terms of money and time, and the redundancy it provides is typically beneficial for correcting only sequencing errors~\cite{press2020hedges}. Thus, designing edit-correcting codes that are practical as inner codes in DNA storage remains an intriguing area for exploration.

Motivated by the need to ensure reliability and reduce read costs in DNA storage, we introduce a binary code designed to correct random edit errors, while also being practical for the typical short lengths in DNA storage. The code construction is inspired by the Guess \& Check (GC) code, initially introduced in~\cite{GCisit}, for correcting only deletions. We improve and generalize the previous code design by integrating novel encoding and decoding strategies to simultaneously correct deletions, insertions, and substitutions.  Our main contributions are as follows:
\begin{itemize}[leftmargin=*]
 \item In Section~\ref{GCC}, we present the GC+ code, a binary, systematic short blocklength code for correcting random edit errors. We detail its encoding and decoding procedures, discuss code properties with respect to parameter selection, and elaborate on the construction components.  
    \item In Section~\ref{sec:res}, we analyze the performance of the GC+ code both theoretically and through simulations. We first derive analytical expressions to evaluate its decoding error probability (frame error rate) over a random channel with independent and identically distributed (i.i.d.) edit errors. We then present simulation results in the short blocklength regime (a few hundred bits) and at moderate-to-high code rates (above~0.5), showing close agreement between the analytical and empirical error probabilities. Furthermore, we show that the code is highly efficient for localized/burst edits, which are also relevant to DNA storage according to recent statistical evidence from experimental data~\cite{tang2023correcting}.
    \item In Section~\ref{DNA}, we demonstrate how the binary GC+ code can be applied to DNA sequences and extend our theoretical and empirical analysis to the quaternary domain. We then compare GC+ with the HEDGES code~\cite{press2020hedges}, highlighting scenarios where GC+ attains significantly lower decoding error probability at higher code rates. Finally, we illustrate the integration of GC+ as an inner code within a DNA storage framework, paired with an outer Reed–Solomon code, enabling reliable data retrieval with minimal read and write costs.
\end{itemize}

\section{Notation} \label{notation}

Let $[n]\triangleq\{1,2,\ldots,n\}$ be the set of integers from $1$ to $n$, and $[i,j]\triangleq\{i,i+1,\ldots,j\}$ denote the set of integers from $i$ to $j\geq i$. Let $\F_q$ be the Galois field of size $q$. Bold letters represent vectors, where lowercase $\x$ denotes a binary vector and uppercase $\X$ denotes a vector in a larger field. We use superscripts to index vectors as $\x^i$ and subscripts to index elements within a vector as $x_i$. For a vector $\x$, $\x_{[i,j]} = (x_i, x_{i+1}, \ldots, x_j)$ represents the substring containing the consecutive bits indexed by $[i,j]$. We use $\langle \x^1,\x^2\rangle$ to refer to the concatenation of two vectors $\x^1$ and~$\x^2$. Let $\boldsymbol{1}^i$ and $\boldsymbol{0}^j$ denote sequences of $i$ consecutive ones and $j$ consecutive zeros, respectively. The $p$-norm of a vector is denoted by $\| \x \|_p$, with $p\geq 1$. Additionally, we use $\|\x\|_0$ to refer to the number of non-zero elements in $\x$. We define $\sign(\alpha)$ as the function that returns the sign of a real number~$\alpha$, with $\sign(0)=+1$ by convention. All logarithms in this paper are of base $2$. Following standard notation, $f(n) = \mathcal{O}(g(n))$ means that $f$ is asymptotically bounded above by $\kappa g(n)$ for some constant $\kappa> 0$. For $a,b \in \mathbb{Z}$, we adopt the following convention for the binomial coefficient $\binom{b}{a}$,
\begin{equation*} \label{eqY}
  \binom{b}{a} = \left\{\def\arraystretch{1.2}%
  \begin{array}{@{}c@{\quad}l@{}}
    1, &  \text{if } b=a,\\
    0, & \text{if } a > b \text{ or } \{a < 0 \text{ and } b \neq a\}, \\
	\frac{b!}{a!(b-a)!}, & \text{otherwise}.\\
  \end{array}\right.
\end{equation*}
We denote the indicator function by \( \mathds{1}_{\{\text{condition}\}} \in \{0,1\} \), which equals $1$ when the condition is true and $0$ otherwise. The probability mass function (PMF) of a random variable $\rv{X}$ (sans-serif) is denoted by $\Pr(\rv{X} = x)$, where we sometimes omit the random variable, e.g., $\Pr(x)$, when it is clear from the context.

\section{GC+ Code} \label{GCC}
\subsection{Encoding}
Consider a binary information message $\u \in \F_2^k$ of length~$k$. Let $\mathsf{Enc}: \F_2^k \to \F_2^n$ be the encoding function that maps the message $\u$ to its corresponding codeword $\x \in \F_2^n$ of length~$n$. The encoding process $\x = \mathsf{Enc}(\u)$ involves the following steps:

\begin{enumerate}[leftmargin=*]
\item The message $\u$ is segmented into $K \triangleq \lceil k/\ell \rceil$ adjacent substrings of length $\ell$ each, denoted by $\u^i\in \F_2^{\ell}$, where $i\in [K]$ and $\u = \langle \u^1, \u^2, \ldots, \u^K \rangle$. Let $U_i\in \F_q$ be the $q$-ary representation of $\u^i\in \F_2^{\ell}$ in $\F_q$, with $q\triangleq 2^{\ell}$, and let $\U\triangleq (U_1,U_2,\ldots,U_K) \in \F_q^K$.\footnote{If the last substring $\u^K$ has length $\ell' < \ell$, the computation of $U_K$ assumes padding with zeros in the $\ell-\ell'$ most significant bit positions.}

\item The sequence $\U\in \F_q^K$ is encoded using an $(N, K)$ systematic Reed-Solomon (RS) code over $\F_q$, with \mbox{$N=K+c\leq q$}, where $c$ is a code parameter representing the number of redundant parity symbols. The resulting sequence is denoted by \mbox{$\X=(X_1, X_2, \ldots, X_{N}) \in \F_q^{N}$}, where $X_{K+1}, X_{K+2}\ldots, X_{N}$ are the parity symbols.

\item Let $\p=\langle \p^1, \p^2, \ldots, \p^c \rangle\in \F_2^{c\ell}$ represent the concatenated binary representation of the $c$ RS parity symbols $X_{K+1}, X_{K+2}, \ldots, X_{N}$. These parity bits undergo additional encoding using a function $f: \F_2^{c\ell} \to \F_2^{c\ell+r_f}$, wherein $f$ introduces a redundancy $r_f$ that enables the detection and/or correction of edit errors in some or all of the parity bits. Specific choices for the function $f$ are discussed in Section~\ref{discF}. The encoded parity bits $f(\p)$ are appended to $\u$ to form $\x$. Thus, the codeword $\x\in \F_2^n$ is given by $$\x = \mathsf{Enc}(\u)=\langle \u, f(\p)\rangle,$$ with $n=k+c\ell+r_f$.
\end{enumerate}

\subsection{Decoding} \label{dec}
Suppose $\x \in \F_2^n$ is affected by edit errors, resulting in a sequence of length $n'$ denoted by $\y \in \F_2^{n'}$. Let $\Delta \triangleq n - n'$ be the number of {\em net} indels in~$\y$. Define $\mathsf{Dec}: \F_2^{n'} \to \F_2^k \cup \{\varnothing\}$ as the decoding function, which is deterministic and either outputs an estimate $\hat{\u} \in \F_2^k$ of the original message or declares a decoding failure, in which case $\hat{\u} = \varnothing$.\footnote{A decoding failure refers to a detectable decoding error, indicating that the decoder acknowledges its inability to decode $\y$.}

The decoding process $\mathsf{Dec}(\y)= \hat{\u}$ employs a guess-and-check mechanism. In this process, a portion of the RS parities is used to generate guesses on $\u$, while the remaining parities are used to check the validity of these guesses. Each of the two parts of the parities serves a different function, denoted as \mbox{$\p_G \triangleq \langle \p^1, \ldots, \p^{c_1} \rangle$} for the first $c_1$ parities used to generate the guesses, and \mbox{$\p_C \triangleq \langle \p^{c_1+1}, \ldots, \p^{c} \rangle$} for the remaining \mbox{$c_2 = c - c_1$} parities used for checking the validity of the guesses. The first step of the GC+ decoder is to leverage the redundancy introduced by the parity encoding function $f$ to either: {\em (i)} Detect that the information bits are error-free and conclude decoding; or {\em (ii)} Retrieve the check parities $\p_C$ and initiate the guess-and-check process, as described next. Further elaboration on strategies to detect errors or retrieve $\p_C$ is deferred to Section~\ref{discF}. The subsequent discussion assumes the successful recovery of the check parities $\p_C$.

The following offers a high-level overview of the guess-and-check process. Let $\y'\triangleq \y_{[1,n'']}$ denote the first \mbox{$n''\triangleq k+c_1\ell+\Delta$} bits of~$\y$. Based on the value of $\Delta$, the decoder hypothesizes the {\em offsets} (i.e., net indels) in each of the $N'\triangleq K+c_1$ segments corresponding to~$\y'$. Subsequently, $\y'$ is segmented according to this hypothesis. Specifically, each substring $i\in[N']$ is segmented to a length of $\ell+\delta_i$, where $\delta_i\in \mathbb{Z}$ is the hypothesized offset. The outcome of this segmentation is decoded using the $(N',K)$ RS code punctured at the last $c_2$ positions, with the decoder taking the $q$-ary representations of segments presumed to have zero offset as input while treating the remaining segments as symbol erasures over~$\F_q$. The output of the RS decoder constitutes a {\em guess} of the information message. If this guess is consistent with the $c_2$ check parities corresponding to $\p_C$, it is validated and the decoder outputs the corresponding binary estimate $\hat{\u}$; otherwise, it proceeds to generate a new guess. If no valid estimate is obtained after all intended guesses have been processed, the decoder declares a decoding failure. A more rigorous description of this process is provided below.

A guess involves assuming a specific {\em offset pattern} \mbox{$\boldsymbol{\delta}=(\delta_1,\ldots,\delta_{N'})\in \Z^{N'}$}, where $\delta_i$ is the number of {\em net} indels in segment $i\in [N']$, and \mbox{$\sum_{i=1}^{N'} \delta_i=\Delta$}.
Given an offset pattern $\boldsymbol{\delta}$, $\y'$ is segmented into $N'$ adjacent binary substrings $\y^1,\ldots,\y^{N'}$, where the length of $\y^{i}$ is $\ell+\delta_i$, for all $i\in [N']$. For a given $\boldsymbol{\delta}$, define $\Y=( Y_1,\ldots,Y_{N'})\in \F_q^{N'}\cup \{?\}$, with
\begin{equation} \label{eqY}
  Y_i = \left\{\def\arraystretch{1.2}%
  \begin{array}{@{}c@{\quad}l@{}}
    (\y^i)_{\F_q}, &  \text{if } \delta_i=0,\\
	?, & \text{otherwise},\\
  \end{array}\right.
\end{equation}
where $?$ denotes an erasure and $(\y^i)_{\F_q}$ is the $q$-ary representation of $\y^i$ in $\F_q$. The sequence $\Y$ is decoded using the punctured $(N',K)$ RS code to obtain a guess $\hat{\U}\in \F_q^K$. Here, we consider syndrome-based RS decoder implementations capable of correcting all combinations of $e$ erasures and $s$ substitutions, provided that \mbox{$e+2s\leq N'-K=c_1$}~\cite{clark, moon2020error}. The decoder then checks whether $\hat{\U}$ is consistent with the $c_2$ parity symbols corresponding to $\p_C$. If the guess is valid, the binary equivalent $\hat{\u}$ of $\hat{\U}$ is returned and decoding is terminated; otherwise, the decoder proceeds to generate a new guess by considering a different offset pattern~$\boldsymbol{\delta}$.

Next, we explicitly define the offset patterns considered during decoding. The decoder performs either a {\em burst check} or a {\em general check}, depending on whether it is configured to correct burst errors in particular or arbitrary errors in general. Each check corresponds to processing a predefined set of offset patterns, as described below.
\subsubsection{Burst check}
The decoder investigates offset patterns associated with scenarios in which all edit errors are localized within $c_1$ consecutive segments. Specifically, it considers the set of patterns $\mathcal{P}_{\text{burst}}$ defined by
\begin{align*}
\mathcal{P}_{\text{burst}}(\Delta,N',c_1) &\triangleq \bigcup_{j=1}^{N'-c_1+1} \mathcal{P}^j(\Delta,N',c_1),  \\
\mathcal{P}^j(\Delta,N',c_1) &\triangleq \bigg\{ \boldsymbol{\delta}\in \Z^{N'} : \sum_{i\in\mathcal{I}_j} \delta_i=\Delta, \delta_i=0~\forall i \notin \mathcal{I}_j  \bigg\},
\end{align*}
where $\mathcal{I}_j\triangleq\{j,j+1,\ldots,j+c_1-1\}$. For a given~$j$, all $c_1$ segments indexed by $\mathcal{I}_j$ are treated as erasures, regardless of the individual values of $\delta_i$ for $i\in \mathcal{I}_j$. Therefore, processing a single pattern from $\mathcal{P}^j$ suffices, as all patterns in $\mathcal{P}^j$ lead to the same guess. Consequently, in the {\em burst check}, the decoder processes a total of $N'-c_1+1 = K+1$ offset patterns.

\subsubsection{General check} \label{gcheck}
The decoder examines the set of offset patterns $\mathcal{P}_{\text{general}}$ defined by 
\begin{multline*}
    \mathcal{P}_{\text{general}}(\Delta,N',c_1,\lambda)\triangleq \bigg\{ \boldsymbol{\delta}\in \Z^{N'} : \sum_{i=1}^{N'} \delta_i=\Delta, \\ \| \boldsymbol{\delta} \|_0 \leq c_1,~ \| \boldsymbol{\delta} \|_1 \leq \left \lvert \Delta \right \rvert + 2\lambda \bigg\},
\end{multline*}
where $\lambda \geq 0$ is a tunable decoding parameter that we call {\em decoding depth}. The set $\mathcal{P}_{\text{general}}$ consists of integer vectors \mbox{\( \boldsymbol{\delta} \in \mathbb{Z}^{N'} \)}whose elements sum to $\Delta$ and that satisfy two constraints: \begin{enumerate*}[label={\textit{(\roman*)}}] \item the number of non-zero elements in $\boldsymbol{\delta}$ is at most $c_1$; \item the L1 norm of $\boldsymbol{\delta}$ is at most $\Delta+2\lambda$. \end{enumerate*} 

The first constraint is imposed since any offset $\delta_i \neq 0$ implies that the corresponding symbol $Y_i$ in~\eqref{eqY} is treated as an erasure, and the $(N',K)$ RS code can correct at most $N'-K=c_1$ erasures. The second constraint, controlled by the decoding depth parameter $\lambda$, serves as a complexity cap by limiting the search space. Specifically, any vector summing to \( \Delta \) must satisfy \( \|\boldsymbol{\delta}\|_1 \geq \left \lvert \Delta \right \rvert \), and the excess $\|\boldsymbol{\delta}\|_1 - \left \lvert \Delta \right \rvert$ reflects the amount of positive-negative cancellations among the entries of $\boldsymbol{\delta}$. By restricting the L1-norm to a maximum of $\left \lvert \Delta \right \rvert+2\lambda$, we limit such cancellations to at most \( \lambda \) units in each direction, thereby capping the search space. Furthermore, the offset patterns in $\mathcal{P}_{\text{general}}$ are processed in increasing order of $\|\boldsymbol{\delta} \|_1$ in order to prioritize scenarios with fewer errors, which are typically more likely in practice.

In summary, the decoder operates as follows:
\begin{enumerate}[leftmargin=*]
\item Use the parity encoding function $f$ either to detect that the information bits are error-free or to retrieve the check parities $\p_C$ and initiate the guess-and-check process.
\item If configured for burst-error correction, perform the {\em burst check}; otherwise, perform the {\em general check}.
\item If no valid estimate $\hat{\u}$ is obtained, declare a decoding failure.
\end{enumerate}

\subsection{Code Properties} \label{secprop}
\subsubsection{Code Rate} 
The redundancy is $n-k=(c_1+c_2)\ell+r_f$, and hence the code rate is $$R=\frac{k}{n}=1-\frac{(c_1+c_2)\ell}{n}-\frac{r_f}{n}.$$




\subsubsection{Time Complexity}
We begin by discussing the encoding complexity of the $q$-ary $(N,K)$ RS code and the decoding complexity of the punctured $(N',K)$ RS code. The encoding complexity of the $(N,K)$ code, utilizing basic polynomial multiplication, is $\O(K(N-K))=\O((c_1+c_2)K)$, and the decoding complexity of the $(N',K)$ code, employing syndrome-based decoding, is $\O(N'(N'-K))=\O(c_1K+c_1^2)$~\cite{moon2020error}. For these encoding and decoding methods, the constants hidden by the $\O$-notation are small, making them efficient for short blocklength codes. For a comprehensive non-asymptotic analysis of these complexities, we refer interested readers to~\cite{chen2008complexity}. 
Furthermore, these complexities are primarily influenced by the number of multiplications in $\F_{2^{\ell}}$; thus, the bit complexities for encoding and decoding include an additional factor of $\O(\log^2(2^{\ell}))=\ell^2$. Recall that \mbox{$K = \lceil k/\ell \rceil$}, resulting in bit encoding and decoding complexities of $\O((c_1+c_2)\ell k)$ and $\O(c_1\ell k+c_1^2\ell^2)$, respectively.

Next, we analyze the encoding and decoding complexities of the GC+ code. We assume that the order of these complexities remains unaffected by the operations related to the parity encoding function $f$. This assumption holds true for all practical purposes, as we later discuss in Section~\ref{discF}. The encoding complexity is dominated by the generation of the $c=c_1+c_2$ RS parities, which is $\O((c_1+c_2)\ell k)$. On the decoding side, the complexity is dominated by the process of generating guesses, computed as the product of the total number of guesses and the complexity of the $(N',K)$ RS decoder. Thus, the worst-case decoding complexity of the {\em burst check} is $(K+1) \O(c_1\ell k+c_1^2\ell^2) = \O(c_1k^2+c_1^2\ell k)$. As for the {\em general check}, the worst-case decoding complexity is given by $\left \lvert \mathcal{P}_{\text{general}}(\Delta,N',c_1,\lambda) \right \rvert \cdot \O(c_1\ell k+c_1^2\ell^2)$, where $\mathcal{P}_{\text{general}}$ is the set of offset patterns defined in Section~\ref{dec}. In Lemma~\ref{prop1}, we provide an expression that determines the exact value of $\left \lvert \mathcal{P}_{\text{general}} \right \rvert$ in terms of $\Delta,N',c_1,$ and $\lambda$, based on combinatorics.  The proof of this lemma is given in Appendix~\ref{appA}.  

\begin{lemma} \label{prop1}
The cardinality of $\mathcal{P}_{\text{general}}(\Delta,N',c_1,\lambda)$ is given by
\begin{equation*} 
     \left \lvert \mathcal{P} \right \rvert = \sum_{i_1=0}^{c_1}  \sum_{i_2=0}^{\lambda}  \sum_{i_3=0}^{\lambda}   \binom{N'}{i_1} \binom{i_1}{i_2}\binom{i_3-1}{i_2-1} \binom{\lvert \Delta \rvert +i_3-1}{i_1-i_2-1}, 
\end{equation*} 
which is upper bounded by
\begin{equation*} 
     \left \lvert \mathcal{P} \right \rvert \le 2^{\,|\Delta|+\lambda-1}\;
i_{1}^{*}\;
\frac{(N')^{i_{1}^{*}}\,(\lambda+i_{1}^{*})^{i_{1}^{*}}}{(i_{1}^{*}!)^{2}},
\end{equation*} 
where $i_{1}^{*}\triangleq \min\{\,c_1,\ \max\{\lambda+1,\ |\Delta|+2\lambda\}\,\}$.
\end{lemma}

It is important to note that since decoding terminates upon finding a valid guess, depending on the underlying random edit error model, the average-case decoding complexity can be significantly lower than the worst-case scenario. The patterns in the {\em general check} are processed in increasing order of $\|\boldsymbol{\delta}\|_1$ to improve the average-case decoding time, assuming that fewer errors are more probable in the underlying channel. 

\subsubsection{Error Correction Capability}
The edit error correction capability of the GC+ decoder at the bit level stems from the erasure and substitution correction capability of the RS code at the $q$-ary level. Specifically, for a given decoder input~$\y'$ and a guess parameterized by an offset pattern $\boldsymbol{\delta}$, the edit errors in $\y'$ can be corrected if $e+2s\leq c_1$, where $e$ and $s$ denote the number of erasures and substitutions, respectively, in $\Y$~(defined in~\eqref{eqY}).  Two sources contribute to the possibility of a decoding error: {\em (i)} A miscorrection, i.e., an {\em undetectable} decoding error, which results from a spurious guess, producing an incorrect estimate $\hat{\u}$ that accidentally aligns with the $c_2$ check parities. {\em (ii)} A decoding failure, i.e., a {\em detectable} decoding error, which occurs when none of the guesses yield a valid estimate, indicating either that the bit-level edit error combination exceeds the error correction capability of the $q$-ary RS code, or that the actual offset pattern was not covered during the guess-and-check process. The probability of a miscorrection is primarily influenced by the value of $c_2$, while the probability of a decoding failure depends on $c_1$, $\lambda$, and the channel parameters. 

The {\em burst check} covers all cases corresponding to scenarios where the edit errors affect at most $c_1$ consecutive segments. This includes cases of burst or localized edit errors, where edit errors of arbitrary type and quantity occur within a window of size $(c_1-1)\ell$. The reasoning extends similarly to the {\em general check}, which covers a broader spectrum of edit error scenarios while leveraging the ability of RS codes to simultaneously correct erasures and substitutions. In Section~\ref{th:analysis}, we provide a theoretical analysis on the overall probability of decoding error of GC+ codes under the random channel model described in Section~\ref{model}. 

\subsection{Choice of the code parameters}\label{sec:choice}
The tunable encoding and decoding parameters of the GC+ code are the segmentation length for encoding $\ell$, number of RS parities for guessing $c_1$, number of RS parities for checking $c_2$, and decoding depth in the {\em general check} $\lambda$. 


The choice of the value of $\ell$ presents a trade-off between the redundancy introduced by the RS code, given by $(c_1+c_2)\ell$, and the decoding complexity which depends on the number of segments $K=\lceil k/\ell \rceil$. Ideally, we seek to minimize $\ell$ to reduce redundancy. However, selecting a small value of $\ell$ poses the following challenges: {\em (i)} The increase in the number of segments $K$ requires the decoder to process a larger number of guesses, leading to increased decoding complexity. {\em (ii)} The RS code imposes the condition $q\geq K+c_1+c_2$, implying $2^{\ell}\geq \lceil k/\ell \rceil + c_1+c_2$, which could be violated for small~$\ell$. A typical choice for $\ell$ is $\ell=\lfloor \log k \rfloor$ as it maintains low redundancy and gives a good trade-off between the code properties; however, higher values of $\ell$ may be considered based on specific application requirements. 

The choice of $c_1$ primarily depends on the desired level of error correction. It can be adapted to the specific error model and application to achieve suitable trade-offs between redundancy and decoding failure rates. The parameter $\lambda$ affects decoding complexity. We typically opt for small values of $\lambda$ to strike a balance between decoding complexity and the probability of decoding failure. The value of $\lambda$ could also be customized according to the number of {\em net} indels $\Delta$, where, for example, instances with $|\Delta|=j$ can be decoded with depth $\lambda_j$. Lastly, the choice of $c_2$ mainly influences the probability of miscorrections. For short blocklength codes, small values of $c_2$ suffice to maintain a very low probability of such errors.

\subsection{Parity encoding function} \label{discF}
As previously mentioned, the purpose of the parity encoding function $f$ is to detect or correct edit errors in the parities $\p =\langle \p_G, \p_C \rangle$. More precisely, we are primarily interested in ``protecting'' the check parities $\p_C$ as their accurate recovery is vital for carrying out the guess-and-check process outlined in Section~\ref{dec}. On the other hand, errors within the guess parities $\p_G$ are implicitly addressed as part of the guess-and-check process. As discussed in the previous section, for typical values of the code parameters such as $\ell=\log k$ and small $c_2$, the length of the check parities $c_2\ell$ is relatively short compared to the information sequence. Thus, we can afford to encode these parities with certain codes that might be inefficient for longer lengths in terms of rate or complexity. Furthermore, for specific types of errors, such as burst or localized errors, it is possible to set up $f$ in a way that the errors impact either the information bits or the parities, but not both. In addition, $f$ can be designed to enable the detection of which of these two cases actually occurred. In the following, we present a non-exhaustive list of possibilities for selecting the parity encoding function $f$.

\subsubsection{Repetition code} The check parities can be encoded using a $t$-repetition code, where each bit is repeated $t$ times, i.e., $f(\p_G,\p_C)=\langle \p_G, \mathsf{rep}_{t}(\p_C)\rangle$. At the decoder, one way to recover the check parities $\p_C$ is to take the majority vote in each window of size $t$ in $\y_{[n'-tc_2\ell,n']}$, where $\y$ is the decoder input defined in Section~\ref{dec}. Then, as previously explained, the guess-and-check process is applied over $\y_{[1,k+c_1\ell+\Delta]}$. The redundancy incurred by the repetition code is $r_f=(t-1)c_2\ell$, resulting in an overall redundancy of $n-k=(c_1+tc_2)\ell$. The impact of the repetition code on the overall encoding and decoding complexity is negligible.

\subsubsection{Exhaustive search} 
Protecting the check parities with significantly less redundancy than repetition codes is possible using methods based on exhaustive search and lookup tables. In general, these methods have exponential time complexity. However, when applied to short inputs, such as the check parities with \mbox{$\ell= \log k$} and small~$c_2$, the complexity remains polynomial in~$k$. One such approach is to construct a code with minimum {\em sequence} Levenshtein distance (SLD), denoted by $d^{\text{SLD}}_{\text{min}}$, via exhaustive search. The so-called SLD is a variant of the standard Levenshtein distance (LD), where the SLD between two strings is defined as the minimum number of deletions, insertions, and substitutions required to make one string a {\em prefix} or {\em suffix} of the other~\cite{buschmann2013levenshtein}. Clearly, $d^{\text{SLD}}_{\text{min}} \leq d^{\text{LD}}_{\text{min}}$, implying that codes designed based on minimum SLD are subject to stricter constraints than those based on LD. 

These additional constraints enable the correction of up to $\tilde{t} \triangleq \left \lfloor (d^{\text{SLD}}_{\text{min}} - 1)/2 \right \rfloor$ edit errors in the check parities, even when their exact boundary within $\boldsymbol{y}$ are unknown. Specifically, suppose the $\tilde{k} = c_2\ell$ check parity bits are encoded using an $(\tilde{n}, \tilde{k})$ code with minimum suffix-SLD $d^{\text{SLD}}_{\text{min}}$, and appended as a suffix to form part of the GC+ codeword. Then, by applying minimum distance decoding to the last $\tilde{n}$ bits of $\boldsymbol{y}$, the decoder can correct up to $\tilde{t}$ edit errors and recover the check parities with zero error, even if these $\tilde{n}$ bits contain ``foreign'' bits not originating from the check parities due to offsets caused by deletions or insertions. 

The resulting redundancy of the GC+ code in this scenario is $n - k = c_1\ell + \tilde{n}$. As for time complexity, the (offline) encoding complexity of this method is \(\mathcal{O}(2^{2\tilde{k}} \cdot \tilde{n}^2)\), since constructing such a code via exhaustive search requires computing the SLD between all codeword pairs, each costing \(\mathcal{O}(\tilde{n}^2)\). This yields \(\mathcal{O}(k^{2c_2} \cdot \tilde{n}^2)\) when \(\tilde{k} = c_2 \log k\). The decoding complexity is \(\mathcal{O}(2^{\tilde{k}} \cdot \tilde{n}^2)\), as the distance between the \(\tilde{n}\) bits and all \(2^{\tilde{k}}\) codewords must be computed for minimum distance decoding, which gives \(\mathcal{O}(k^{c_2} \cdot \tilde{n}^2)\) under the same assumptions.


\subsubsection{Buffer} \label{buff}
Suppose that the edit errors are localized within any window of $w<n$ consecutive bit positions, where the location of the window is unknown to the decoder but its size is known. Then, it is possible to insert a buffer between the information and parity bits to achieve the following: {\em (i)}~Ensure that edit errors cannot simultaneously affect the information and parity bits. {\em (ii)}~Enable the detection of whether edit errors have affected the information bits or not in all cases where $\Delta \neq 0$.  Recent works~\cite{hanna2021codes, bitar2021optimal, tang2023correcting} have explored such buffers. Following~\cite[Lemma 17]{tang2023correcting}, the buffer \mbox{$\b\triangleq \langle \boldsymbol{1}^{w+1}, \boldsymbol{0}^{w+1}, \boldsymbol{1}^{w+1}\rangle$} of length \mbox{$3(w+1)$} serves the intended purpose. The parity encoding function is then defined as \mbox{$f(\p_G, \p_C) = \langle \b, \p_G, \p_C \rangle$}, and the decoder uses this buffer to either directly output the error-free information bits or, if they are affected, rely on the error-free parities to initiate the {\em burst check}. Here, decoding is simplified by fixing $\p_G$ throughout the guess-and-check process and operating over $K$ segments instead of $K + c_1$. Note that the buffer is ineffective when $\Delta = 0$, in which case the decoder discards the buffer positions in $\boldsymbol{y}$ and passes the remaining bits to the RS decoder to correct substitutions. The overall redundancy in this scenario is $n - k = (c_1 + c_2)\ell + 3(w+1)$, and the impact of $f$ on the time complexity of the GC+ code is negligible. 

\section{Performance Analysis} \label{sec:res}
\subsection{Channel Model} \label{model} 
In this section, we evaluate the error-correction capability of the GC+ code under the following channel model. Let $\x \in \Sigma^{n}$ and $\y \in \Sigma^{n'}$ be the input and output of the channel, respectively, where $\Sigma$ denotes the alphabet. Each symbol in $\x_{[i,i+w-1]}$ is independently edited with probability $P_{\text{edit}}$, where $i$ is sampled uniformly at random from $\{1,2,\ldots,n-w+1\}$ and $w\leq n$ is a channel parameter representing the size of the window where the errors are localized. All symbols in $\x_{[1,i-1]}$ and $\x_{[i+w,n]}$ are retained. The probabilities of deletion, insertion, and substitution are denoted by $P_d$, $P_i$, and $P_s$, respectively, with \mbox{$P_{\text{edit}}=P_d+P_i+P_s$}. For an input symbol~$x$, the output symbol $y$ of the channel is determined as follows: if $x$ is deleted, the output $y$ is an empty string; if $x$ experiences an insertion, then $y=\langle \sigma, x \rangle$ where $\sigma$ is chosen uniformly at random from $\Sigma$; if $x$ is substituted, then $y=\Tilde{x}$, where $\Tilde{x}$ is chosen uniformly at random from $\Sigma \setminus x$. Thus, the channel is characterized by the parameters $(w,n,P_{\text{edit}},P_d,P_i,P_s)$, and under this model, we have $|n-n'|\leq n$. The average edit rate introduced by this channel, denoted by $\varepsilon_{av}$, is given by $\varepsilon_{av}=P_{\text{edit}}\times \frac{w}{n}$.

Note that the case $w=n$ reduces to the scenario of i.i.d. edits, with $\varepsilon_{av}=P_{\text{edit}}$. In our theoretical analysis in Section~\ref{th:analysis}, we focus on this case to analytically evaluate the decoding performance of the GC+ code under arbitrarily located edits. In our simulation results in Section~\ref{simul}, we study both scenarios: the case $w<n$ with high values of $P_{\text{edit}}$ to simulate burst or localized edits, and the case $w=n$ with lower values of $P_{\text{edit}}$ to model scenarios with edits distributed arbitrarily over the entire sequence. 

\subsection{Theoretical Analysis} \label{th:analysis}
\subsubsection{Definitions and preliminaries} \label{defprem}
Recall that both decoding failures and miscorrections contribute to the overall probability of decoding error. Thus, for a given information message $\boldsymbol{u} \in \mathbb{F}_2^k$, we define the decoding error event as 
\begin{equation}
\mathcal{E} \triangleq \{\hat{\Rv{U}} = \varnothing \} \cup \{\hat{\Rv{U}} \neq \boldsymbol{u} \land \hat{\Rv{U}} \neq \varnothing\},
\end{equation}
where $\hat{\Rv{U}}\in \F_2^k \cup \{\varnothing\}$ is the random variable representing the estimated message. Since the GC+ decoder is deterministic, the randomness in $\hat{\Rv{U}}$ for a given $\u$ arises solely from the channel model described in Section~\ref{model}.

Let \( \mathsf{N}^{\mathsf{del}}_i \), \( \mathsf{N}^{\mathsf{ins}}_i \), and \( \mathsf{N}^{\mathsf{sub}}_i \) denote the random variables representing the respective number of deletions, insertions, and substitutions in segment \( i \in [N'] \) of the GC\texttt{+} codeword, where \( N' = K + c_1 \) is the total number of segments of length \( \ell \) bits to which the guess-and-check process is applied. According to the channel model, we have $\mathsf{N}^{\mathsf{del}}_i + \mathsf{N}^{\mathsf{ins}}_i + \mathsf{N}^{\mathsf{sub}}_i \in [\ell]$.

Next, we define three “problematic” events, each of which is highly likely to cause a decoding error if realized during the {\em general check} decoding step described in Section~\ref{gcheck}.%
%

For each $i \in [N']$, we define the random variable
\begin{equation} \label{defE}
\rv{E}_i \triangleq 
\begin{cases}
0 & \text{if } \mathsf{N}^{\mathsf{del}}_i = \mathsf{N}^{\mathsf{ins}}_i = \mathsf{N}^{\mathsf{sub}}_i = 0, \\
1 & \text{if } \mathsf{N}^{\mathsf{del}}_i \neq \mathsf{N}^{\mathsf{ins}}_i , \\
2 & \text{otherwise},
\end{cases}
\end{equation}
and let $\tilde{\rv{E}} \triangleq \sum_{i=1}^{N'} \rv{E}_i$. Here, $\rv{E}_i=0$ indicates that segment~$i$ is error-free, $\rv{E}_i=1$ indicates a non-zero offset in segment~$i$, and $\rv{E}_i=2$ indicates the presence of errors in segment~$i$ that do not induce any offset. From the perspective of the GC+ decoder, these three scenarios correspond, respectively, to $Y_i$ in the sequence $\Y$ defined in~\eqref{eqY} being error-free ($\rv{E}_i=0$), erased ($\rv{E}_i=1$), or substituted ($\rv{E}_i=2$). Since the underlying RS code can correct $e$ erasures and $s$ substitutions in~$\Y$ if and only if \mbox{$e+2s\leq c_1$}, where $c_1$ is the number of guess parities, error patterns satisfying $\tilde{\rv{E}} > c_1$ lie beyond the error correction capability of the RS code. Therefore, we define the first problematic event as 
\begin{equation}
\mathcal{E}_1 \triangleq \{\tilde{\rv{E}} > c_1\}.
\end{equation}

Let $\Rv{D}=(\rv{D}_1,\rv{D}_2,\ldots,\rv{D}_{N'})$ be the random vector representing the segment-level offsets, where 
\begin{equation}
\mathsf{D}_i \triangleq \mathsf{N}^{\mathsf{ins}}_i - \mathsf{N}^{\mathsf{del}}_i \in [-\ell, \ell],
\end{equation}
and let $\tilde{\rv{D}} \triangleq \sum_{i=1}^{N'} \rv{D}_i$. As previously explained in Section~\ref{gcheck}, the offset patterns examined by the GC+ decoder during the general check satisfy the constraint $\| \boldsymbol{\delta} \|_1 \leq \left \lvert \Delta \right \rvert + 2\lambda$, where $\lambda$ is the decoding depth parameter. Hence, if the actual offset pattern induced by the channel violates this constraint, it will fall outside the search space of the GC+ decoder. We therefore define the second problematic event as
\begin{equation}
\mathcal{E}_2 \triangleq \big\{ \| \Rv{D} \|_1 >  \lvert \tilde{\rv{D}}  \rvert + 2\lambda(\tilde{\rv{D}}) \big\},
\end{equation}
where the dependence of $\lambda$ on $\tilde{\rv{D}}$ captures the general case discussed in Section~\ref{sec:choice}, in which different decoding depths may be considered for each realization $\Delta$ of $\tilde{\rv{D}}$.

Recall that the first decoding step involves recovering the check parities, denoted by $\p_C$, based on the parity encoding function. Let $\hat{\Rv{P}}_C$ be the random variable representing the estimated check parities. As previously explained, accurate recovery of the check parities is crucial for the guess-and-check process, as all guesses are validated or discarded based on $\p_C$. Hence, for a given  $\p_C$, we define the third problematic event as
 \begin{equation}
 \mathcal{E}_3 \triangleq \{ \hat{\Rv{P}}_C \neq \p_C \}.
 \end{equation}

\subsubsection{Main result}
Theorem~\ref{thm2} bounds the decoding error probability of the GC+ code in terms of the three events, $\mathcal{E}_1$, $\mathcal{E}_2$, and $\mathcal{E}_3$. Although the inequality follows from a standard union bound argument, its significance lies in the decomposition it provides: the overall error probability is broken down into components tied to distinct parts of the decoding process. This formulation serves as the basis for the subsequent analysis, where each event probability is evaluated separately (Sections~\ref{err1}–\ref{err3}) and later combined with numerical results to demonstrate the effectiveness of the bound.

\begin{theorem} \label{thm2}
The probability of decoding error of the GC+ code satisfies
\begin{equation*} \label{eq5}
\Pr(\mathcal{E}) \leq \Pr(\mathcal{E}_1) + \Pr(\mathcal{E}_2) + \Pr(\mathcal{E}_3) + \Pr(\mathcal{E} \mid \bar{\mathcal{E}}_1\cap \bar{\mathcal{E}}_2\cap \bar{\mathcal{E}}_3).
\end{equation*}
\end{theorem}
\begin{proof}
Let $\mathcal{A}\triangleq \{ \mathcal{E}_1 \cup \mathcal{E}_2 \cup \mathcal{E}_3\}$. Then, by the law of total probability, we have
\begin{equation} \label{eqp1}
\Pr(\mathcal{E}) = \Pr(\mathcal{E}\cap \mathcal{A}) + \Pr(\mathcal{E}\cap \bar{\mathcal{A}}).
\end{equation}
Since $\mathcal{E} \cap \mathcal{A} = (\mathcal{E}\cap \mathcal{E}_1) \cup (\mathcal{E}\cap \mathcal{E}_2) \cup (\mathcal{E}\cap \mathcal{E}_3)$, the union bound gives
\begin{align} \label{eqp2}
\Pr(\mathcal{E}\cap \mathcal{A}) \leq \sum_{j=1}^3 \Pr(\mathcal{E}\cap \mathcal{E}_j) = \sum_{j=1}^3 \Pr(\mathcal{E} \mid \mathcal{E}_j) \Pr(\mathcal{E}_j).
\end{align}
Since the events $\mathcal{E}_j$, $j=1,2,3$, are problematic for decoding due to the aforementioned reasons, we expect the trivial bound $\Pr(\mathcal{E} \mid \mathcal{E}_j)\leq 1$ to be effective. Hence, we write
\begin{align} \label{eqp3}
\Pr(\mathcal{E}\cap \mathcal{A}) \leq \Pr(\mathcal{E}_1) + \Pr(\mathcal{E}_2) + \Pr(\mathcal{E}_3).
\end{align}
Since $\mathcal{E} \cap \bar{\mathcal{A}} = \mathcal{E}\cap \bar{\mathcal{E}}_1 \cap \bar{\mathcal{E}}_2 \cap \bar{\mathcal{E}}_3$, we also have
\begin{equation} \label{eqp4}
\Pr(\mathcal{E}\cap \bar{\mathcal{A}}) = \Pr(\mathcal{E} \mid \bar{\mathcal{E}}_1\cap \bar{\mathcal{E}}_2\cap \bar{\mathcal{E}}_3) \Pr(\bar{\mathcal{E}}_1\cap \bar{\mathcal{E}}_2\cap \bar{\mathcal{E}}_3). 
\end{equation}
The result in Theorem~\ref{thm2} therefore follows from~\eqref{eqp1}, \eqref{eqp3}, \eqref{eqp4}, and the fact that $\Pr(\bar{\mathcal{E}}_1\cap \bar{\mathcal{E}}_2\cap \bar{\mathcal{E}}_3) \leq 1$. 
\end{proof}
The event $\bar{\mathcal{E}}_1 \cap \bar{\mathcal{E}}_2 \cap \bar{\mathcal{E}}_3$ in the conditional probability term of Theorem~\ref{thm2} corresponds to scenarios where: {\em (i)}~the check parities are correctly recovered, {\em (ii)}~the actual offset pattern induced by the channel lies within the search space of the GC+ decoder, and {\em (iii)}~the number of erasures and substitutions in $\Y$ are within the error correction capability of the RS code. Intuitively, decoding errors are unlikely when these three events occur simultaneously. However, a non-zero probability of decoding error remains due to spurious guesses that may coincidentally match with the check parities. This probability can nonetheless be made arbitrarily small by choosing a sufficiently large number of check parities $c_2$ to avoid, with high probability, such spurious matches.\footnote{For short blocklengths, small values of $c_2$, such as $c_2=1$ or $2$ are sufficient to avoid such spurious matches with high probability.} Therefore, for all practical purposes, we assume that the term $\Pr(\mathcal{E} \mid \bar{\mathcal{E}}_1\cap \bar{\mathcal{E}}_2\cap \bar{\mathcal{E}}_3)$ in Theorem~\ref{thm2} is negligible and write
\begin{equation} \label{eq6}
\Pr(\mathcal{E}) \lessapprox \Pr(\mathcal{E}_1) + \Pr(\mathcal{E}_2) + \Pr(\mathcal{E}_3).
\end{equation}
In what follows, we analyze the probabilities of the events $\mathcal{E}_1$, $\mathcal{E}_2$, and $\mathcal{E}_3$ to evaluate~\eqref{eq6}. In Section~\ref{simul}, we provide numerical results demonstrating that the empirical probability of decoding error of the GC+ code can be closely approximated analytically using~\eqref{eq6}.

\subsubsection{Probability of event $\mathcal{E}_1$} \label{err1}
It follows from the channel model and~\eqref{defE} that $\rv{E}_1, \rv{E}_2,\ldots,\rv{E}_{N'}$ are i.i.d. random variables. Define $\alpha_0 \triangleq \Pr(\rv{E}_i = 0)$, $\alpha_1 \triangleq \Pr(\rv{E}_i = 1)$, and \mbox{$\alpha_2 \triangleq \Pr(\rv{E}_i = 2)$}. Lemma~\ref{lemm2} expresses the probability of \mbox{$\mathcal{E}_1 \;=\;\{\tilde{\rv{E}} > c_1\}$} in terms of the code and channel parameters.
\begin{lemma} \label{lemm2}
The probability of $\mathcal{E}_1$ is given by
\[
\Pr(\mathcal{E}_1) = 1-\hspace{-0.4cm} \sum_{j_1+2j_2\leq c_1} \hspace{-0.15cm} \binom{N'}{j_1,j_2,N'-j_1-j_2} \alpha_1^{j_1}  \alpha_2^{j_2} \alpha_0^{N'-j_1-j_2}, 
\]
where $j_{1}\in \big[0, \, c_1\big]$, \mbox{$j_{2}\in \big[0,\, \left \lfloor (c_{1}-j_{1})/2 \right \rfloor \big]$}, and
\begin{align*}
\alpha_0 &= \left(1-P_d-P_i-P_s\right)^{\ell}, \\
\alpha_1 &= 1 - \sum_{j=0}^{\left \lfloor \ell/2 \right \rfloor} \binom{\ell}{j,j,\ell-2j} (P_d)^j (P_i)^j (1-P_d-P_i)^{\ell-2j}, \\
\alpha_{2} &= 1-\alpha_{0}-\alpha_{1},
\end{align*}
\end{lemma}
\begin{proof}
The event $\{\rv{E}_i=0\}$ corresponds to the case where no errors occur in a segment of length $\ell$. Since the channel introduces i.i.d. edit errors, it follows that $\alpha_0=(1-P_{\text{edit}})^{\ell}$, where $P_{\text{edit}}=P_d+P_i+P_s$. For the event $\{\rv{E}_i=1\}$, we have
$$\alpha_1=\Pr(\mathsf{N}^{\mathsf{del}}_i \neq \mathsf{N}^{\mathsf{ins}}_i)=1-\Pr(\mathsf{N}^{\mathsf{del}}_i = \mathsf{N}^{\mathsf{ins}}_i).$$
From the channel model, the vector $(\mathsf{N}^{\mathsf{del}}_i, \mathsf{N}^{\mathsf{ins}}_i, \mathsf{N}^{\mathsf{sub}}_i)$ follows a multinomial distribution with parameters $(\ell\,;\,P_d,P_i,P_s)$. Therefore, the probability that the number of deletions equals the number of insertions is obtained by summing the multinomial probabilities over all realizations \((j,j,\ell-2j)\) for $j \in [0, \lfloor \ell/2 \rfloor]$. Subtracting this sum from $1$ gives $\alpha_1$, and consequently, $\alpha_2=1-\alpha_0-\alpha_1$.

To evaluate the probability of $\mathcal{E}_1$, define the count vector $$(\rv{J}_0,\rv{J}_1,\rv{J}_2)
 = \bigl(\left \lvert \{i:\rv{E}_i=0\} \right \rvert,
        \left \lvert \{i:\rv{E}_i=1\} \right \rvert,
        \left \lvert \{i:\rv{E}_i=2\} \right \rvert \bigr), $$
so that $\tilde{\rv{E}}=\rv{J}_1+2\rv{J}_2$ and $\rv{J}_0+\rv{J}_1+\rv{J}_2=N'$. Because $\rv{E}_1, \rv{E}_2,\ldots,\rv{E}_{N'}$ are i.i.d., the vector $(\rv{J}_0,\rv{J}_1,\rv{J}_2)$ follows a multinomial distribution with parameters $(N'\,;\,\alpha_0,\alpha_1,\alpha_2)$. Summing the multinomial probabilities over all realizations \((j_0,j_1,j_2)\) that satisfy $j_1+2j_2\leq c_1$ and $j_0+j_1+j_2= N'$ yields
\(\Pr(\tilde{\rv{E}}\le c_1)\), and subtracting from~\(1\) gives $\Pr(\mathcal{E}_1)$.
\end{proof}


\begin{figure*}[h!]
\centering
\begin{subfigure}[b]{0.32\textwidth}
\centering
\begin{tikzpicture}[scale=0.7]
\begin{semilogyaxis}[
    legend cell align={left},
    tick scale binop=\times, 
    ylabel style={yshift=-0.5ex},
    xlabel={Edit Probability $P_{\text{edit}}$ },
    ylabel={Probability of Decoding Error},
    grid=both,
          minor tick num=1,
    xtick={0.001, 0.003, 0.005, 0.007, 0.009, 0.011, 0.013, 0.015},
    ytick={0.00001,0.0001, 0.001, 0.01, 0.1, 1},
    legend pos=south east,
    legend style={font=\scriptsize},
    xmajorgrids=true,
    xminorgrids=true,
    ymajorgrids=true,
    yminorgrids=true,
    major grid style={black!15}, 
    minor grid style={dotted,black!40}, 
    ymax=1,
    ymin=1e-5,
        xmin=0.001,
    xmax=0.015,
]

\addplot+[color=blue,mark=o,mark size=1.5, mark options={fill=none},line width=0.8pt] coordinates {
(0.001,4.8000e-05)
(0.002,2.4100e-04)
(0.003,7.7800e-04)
(0.004,1.6860e-03)
(0.005,3.3580e-03)
(0.006,5.9430e-03)
(0.007,9.9550e-03)
(0.008,1.5197e-02)
(0.009,2.2120e-02)
(0.010,3.0270e-02)
(0.011,4.1420e-02)
(0.012,5.6270e-02)
(0.013,7.1960e-02)
(0.014,8.9380e-02)
(0.015,1.0753e-01)
%
};

\addplot+[color=blue,mark=none,mark size=1.5, mark options={solid, fill=none,rotate=180},line width=0.8pt,dashed] coordinates {
(0.001,7.62068e-05)
(0.002,0.000392682)
(0.003,0.00117487)
(0.004,0.00272703)
(0.005,0.00538953)
(0.006,0.00950926)
(0.007,0.0154195)
(0.008,0.0234263)
(0.009,0.0337993)
(0.010,0.0467665)
(0.011,0.0625105)
(0.012,0.0811673)
(0.013,0.102826)
(0.014,0.127531)
(0.015,0.15528)
};

\addplot+[color=red,mark=o,mark size=1.5, mark options={fill=none},line width=0.7pt] coordinates {
(0.001,4.700e-05)
(0.002,2.100e-04)
(0.003,5.290e-04)
(0.004,1.075e-03)
(0.005,1.876e-03)
(0.006,3.326e-03)
(0.007,5.364e-03)
(0.008,8.371e-03)
(0.009,1.225e-02)
(0.010,1.767e-02)
(0.011,2.467e-02)
(0.012,3.411e-02)
(0.013,4.581e-02)
(0.014,5.955e-02)
(0.015,7.410e-02)
%
};

\addplot+[color=red,mark=none,mark size=1.5, mark options={solid, fill=none,rotate=180},line width=0.7pt,dashed] coordinates {
(0.001,7.61196e-05)
(0.002,0.000316833)
(0.003,0.00077548)
(0.004,0.00156143)
(0.005,0.00284109)
(0.006,0.00482992)
(0.007,0.00777932)
(0.008,0.0119612)
(0.009,0.0176528)
(0.010,0.0251218)
(0.011,0.0346149)
(0.012,0.0463475)
(0.013,0.0604967)
(0.014,0.077197)
(0.015,0.096538)
};

\addplot+[color=orange,mark=none,mark size=1.5, mark options={solid, fill=none,rotate=180},line width=1pt,densely dashdotted] coordinates {
(0.001,0.195157)
(0.002,0.352369)
(0.003,0.478986)
(0.004,0.580940)
(0.005,0.663017)
(0.006,0.729077)
(0.007,0.782236)
(0.008,0.825002)
(0.009,0.859401)
(0.010,0.887063)
(0.011,0.909303)
(0.012,0.927179)
(0.013,0.941545)
(0.014,0.953087)
(0.015,0.962359)
};

\legend{}
\addlegendentry{GC+ Empirical (sym.)}
\addlegendentry{GC+ Theoretical (sym.)}
\addlegendentry{GC+ Empirical (asym.)}
\addlegendentry{GC+ Theoretical (asym.)}
\addlegendentry{Uncoded}

\end{semilogyaxis}
\end{tikzpicture}
\caption{Repetition code ($t=3$); $(n, R) = (217, 0.645)$.}
\label{suba}
\end{subfigure}\hfill
\begin{subfigure}[b]{0.32\textwidth}
\centering
\begin{tikzpicture}[scale=0.7]
\begin{semilogyaxis}[
    legend cell align={left},
    tick scale binop=\times, 
    ylabel style={yshift=-0.5ex},
    xlabel={Edit Probability $P_{\text{edit}}$ },
    ylabel={Probability of Decoding Error},
    grid=both,
    minor tick num=1,
    xtick={0.001, 0.003, 0.005, 0.007, 0.009, 0.011, 0.013, 0.015},
    ytick={1e-6, 0.00001,0.0001, 0.001, 0.01, 0.1, 1},
    legend pos=south east,
    legend style={font=\scriptsize},
    xmajorgrids=true,
    xminorgrids=true,
    ymajorgrids=true,
    yminorgrids=true,
    major grid style={black!15}, 
    minor grid style={dotted,black!40}, 
    ymax=1,
    ymin=1e-6,
    xmin=0.001,
    xmax=0.015,
]

\addplot+[color=blue,mark=o,mark size=1.5, mark options={fill=none},line width=0.8pt] coordinates {
(0.001,5.0000e-06)
(0.002,7.8000e-05)
(0.003,3.9400e-04)
(0.004,1.1230e-03)
(0.005,2.5170e-03)
(0.006,5.0680e-03)
(0.007,8.7110e-03)
(0.008,1.3729e-02)
(0.009,2.0510e-02)
(0.010,3.0180e-02)
(0.011,4.1580e-02)
(0.012,5.5860e-02)
(0.013,7.0970e-02)
(0.014,8.7450e-02)
(0.015,1.0673e-01)
%
%
};

\addplot+[color=blue,mark=none,mark size=1.5, mark options={solid, fill=none,rotate=180},line width=0.8pt,dashed] coordinates {
(0.001,8.94771e-06)
(0.002,0.000127642)
(0.003,0.000587342)
(0.004,0.00169792)
(0.005,0.0038051)
(0.006,0.00726097)
(0.007,0.0124038)
(0.008,0.0195443)
(0.009,0.0289569)
(0.010,0.0408738)
(0.011,0.055482)
(0.012,0.0729215)
(0.013,0.0932856)
(0.014,0.116621)
(0.015,0.142932)
};

\addplot+[color=red,mark=o,mark size=1.5, mark options={fill=none},line width=0.8pt] coordinates {
(0.001,1e-6)
(0.002,1.000e-05)
(0.003,6.900e-05)
(0.004,2.810e-04)
(0.005,7.230e-04)
(0.006,1.657e-03)
(0.007,3.088e-03)
(0.008,5.587e-03)
(0.009,9.180e-03)
(0.010,1.367e-02)
(0.011,2.067e-02)
(0.012,2.924e-02)
(0.013,3.958e-02)
(0.014,5.220e-02)
(0.015,6.697e-02)
%
};

\addplot+[color=red,mark=none,mark size=1.5, mark options={solid, fill=none,rotate=180},line width=0.8pt,dashed] coordinates {
(0.001,1.49808e-06)
(0.002,2.26707e-05)
(0.003,0.000123196)
(0.004,0.000418587)
(0.005,0.00108121)
(0.006,0.00233231)
(0.007,0.00442886)
(0.008,0.00764826)
(0.009,0.0122728)
(0.010,0.0185756)
(0.011,0.026808)
(0.012,0.0371902)
(0.013,0.0499041)
(0.014,0.0650885)
(0.015,0.0828374)
};

\addplot+[color=orange,mark=none,mark size=1.5, mark options={solid, fill=none,rotate=180},line width=1pt,densely dashdotted] coordinates {
(0.001,0.206352)
(0.002,0.370269)
(0.003,0.500447)
(0.004,0.603807)
(0.005,0.685854)
(0.006,0.750968)
(0.007,0.802632)
(0.008,0.843615)
(0.009,0.876116)
(0.010,0.901886)
(0.011,0.922314)
(0.012,0.938503)
(0.013,0.951330)
(0.014,0.961491)
(0.015,0.969537)
};

\legend{}
\addlegendentry{GC+ Empirical (sym.)}
\addlegendentry{GC+ Theoretical (sym.)}
\addlegendentry{GC+ Empirical (asym.)}
\addlegendentry{GC+ Theoretical (asym.)}
\addlegendentry{Uncoded}

\end{semilogyaxis}
\end{tikzpicture}
\caption{Repetition code ($t=5$); $(n, R) = (231,0.606)$.}
\label{subb}
\end{subfigure}\hfill
\begin{subfigure}[b]{0.34\textwidth}
\centering
\begin{tikzpicture}[scale=0.7]
\begin{semilogyaxis}[
    legend cell align={left},
    tick scale binop=\times, 
    ylabel style={yshift=-0.5ex},
    xlabel={Edit Probability $P_{\text{edit}}$ },
    ylabel={Probability of Decoding Error},
    grid=both,
    minor tick num=1,
    xtick={0.001, 0.003, 0.005, 0.007, 0.009, 0.011, 0.013, 0.015},
    ytick={1e-6, 0.00001,0.0001, 0.001, 0.01, 0.1, 1},
    legend pos=south east,
    legend style={font=\scriptsize},
    xmajorgrids=true,
    xminorgrids=true,
    ymajorgrids=true,
    yminorgrids=true,
    major grid style={black!15}, 
    minor grid style={dotted,black!40}, 
    ymax=1,
    ymin=1e-6,
        xmin=0.001,
    xmax=0.015,
]

\addplot+[color=blue,mark=o,mark size=1.5, mark options={fill=none},line width=0.8pt] coordinates {
(0.001,7.0000e-06)
(0.002,6.3000e-05)
(0.003,3.5700e-04)
(0.004,1.0090e-03)
(0.005,2.2820e-03)
(0.006,4.4000e-03)
(0.007,7.9080e-03)
(0.008,1.2633e-02)
(0.009,1.8260e-02)
(0.010,2.7550e-02)
(0.011,3.7910e-02)
(0.012,4.8730e-02)
(0.013,6.3050e-02)
(0.014,8.0720e-02)
(0.015,9.8740e-02)
%
};

\addplot+[color=blue,mark=none,mark size=1.5, mark options={solid, fill=none,rotate=180},line width=0.8pt,dashed] coordinates {
(0.001,9.57582e-06)
(0.002,0.00013261)
(0.003,0.000603922)
(0.004,0.00173678)
(0.005,0.00388013)
(0.006,0.00738914)
(0.007,0.012605)
(0.008,0.0198412)
(0.009,0.0293747)
(0.010,0.0414404)
(0.011,0.0562275)
(0.012,0.0738781)
(0.013,0.0944876)
(0.014,0.118105)
(0.015,0.144736)
};

\addplot+[color=red,mark=o,mark size=1.5, mark options={fill=none},line width=0.7pt] coordinates {
(0.001,1.000e-06)
(0.002,1.200e-05)
(0.003,7.200e-05)
(0.004,2.310e-04)
(0.005,6.590e-04)
(0.006,1.526e-03)
(0.007,2.847e-03)
(0.008,5.096e-03)
(0.009,8.310e-03)
(0.010,1.353e-02)
(0.011,1.987e-02)
(0.012,2.757e-02)
(0.013,3.775e-02)
(0.014,5.018e-02)
(0.015,6.415e-02)
%
};

\addplot+[color=red,mark=none,mark size=1.5, mark options={solid, fill=none,rotate=180},line width=0.7pt,dashed] coordinates {
(0.001,1.89931e-06)
(0.002,2.58233e-05)
(0.003,0.000133646)
(0.004,0.000442911)
(0.005,0.00112787)
(0.006,0.00241147)
(0.007,0.00455227)
(0.008,0.00782912)
(0.009,0.0125256)
(0.010,0.018916)
(0.011,0.0272527)
(0.012,0.0377568)
(0.013,0.0506111)
(0.014,0.0659551)
(0.015,0.0838832)
};

\addplot+[color=orange,mark=none,mark size=1.5, mark options={solid, fill=none,rotate=180},line width=1pt,densely dashdotted] coordinates {
(0.001,0.194352)
(0.002,0.351071)
(0.003,0.477418)
(0.004,0.579257)
(0.005,0.661323)
(0.006,0.727442)
(0.007,0.780700)
(0.008,0.823591)
(0.009,0.858124)
(0.010,0.885922)
(0.011,0.908294)
(0.012,0.926294)
(0.013,0.940775)
(0.014,0.952421)
(0.015,0.961786)
};

\legend{}
\addlegendentry{GC+ Empirical (sym.)}
\addlegendentry{GC+ Theoretical (sym.)}
\addlegendentry{GC+ Empirical (asym.)}
\addlegendentry{GC+ Theoretical (asym.)}
\addlegendentry{Uncoded}

\end{semilogyaxis}
\end{tikzpicture}
\caption{Exhaustive search ($d^{\text{SLD}}_{\text{min}}=5$); $(n,R)=(216,0.648)$.}
\label{subc}
\end{subfigure}
\caption{Empirical and theoretical probability of decoding error (frame error rate) of the binary GC+ code under i.i.d. edit errors. The message length is $k=140$, and the code parameters are set to $\ell=\lfloor \log k \rfloor = 7$, $(c_1,c_2)=(8,1)$, with $\lambda=1$ for $|\Delta|\leq 1$ and $\lambda=0$ otherwise. Subfigures~(a)–(c) show results for different configurations of the parity encoding function (Section~\ref{discF}), where~(a) and (b) both use $t$-fold repetition codes with different $t$, and~(c) uses a code of length $20$ with $d^{\text{SLD}}_{\min}=5$ constructed via exhaustive search. The corresponding blocklengths $n$ and code rates $R$ of the GC+ code are given in the subcaptions. Theoretical values follow from~\eqref{eq6}, and empirical ones are averaged over $10^6$ independent simulation runs. For reference, the uncoded frame error rate, given by $1-(1-P_{\text{edit}})^n$, is also plotted. ``Sym'' and ``asym.'' denote the symmetric and asymmetric edit error proportions specified in the text.}
\label{fig2}
\end{figure*}

\subsubsection{Probability of event $\mathcal{E}_2$}\label{err2}
Recall that $\tilde{\rv{D}} = \sum_{i=1}^{N'} \rv{D}_i$ denotes the cumulative offset. Let 
\[
\tilde{\rv{D}}^+ \triangleq \sum_{i=1}^{N'} \rv{D}_i\, \mathds{1}_{\{\rv{D}_i>0\}}, \quad \tilde{\rv{D}}^- \triangleq -\sum_{i=1}^{N'} \rv{D}_i\, \mathds{1}_{\{\rv{D}_i<0\}},
\]
and define 
\[
\tilde{\rv{Z}}\triangleq \min\{\tilde{\rv{D}}^+,  \tilde{\rv{D}}^-\} \in [0,\ell N'].
\]
When $\tilde{\rv{D}}\neq 0$, $\tilde{\rv{Z}}$ equals the total magnitude of the offsets whose sign is opposite to that of $\tilde{\rv{D}}$.
When $\tilde{\rv{D}}= 0$, the positive and negative parts have equal magnitude, and $\tilde{\rv{Z}}$ is equal to this common value. Lemma~\ref{probE2} expresses the probability of event $\mathcal{E}_2$ in terms of the joint distribution of $(\tilde{\rv{Z}}, \tilde{\rv{D}})$.
\begin{lemma} \label{probE2}
Let $\lambda : [-\ell N', \ell N'] \to [0, \ell N']$ be the fixed function representing the predefined decoding depth parameter for each value of the cumulative offset $\Delta$. Then, the probability of event $\mathcal{E}_2$ is given by
\[
  \Pr(\mathcal{E}_2)
      = 1-\sum_{\Delta=-\ell N}^{\ell N}
             \sum_{z=0}^{\lambda(\Delta)}
                 \Pr(\tilde{\rv{Z}}=z,\tilde{\rv{D}}=\Delta).
\]
\end{lemma}
\begin{proof}
The event $\mathcal{E}_2=\Bigl\{\|\Rv D\|_1>\lvert\tilde{\rv D}\rvert+2\lambda(\tilde{\rv D})\Bigr\}$ can be rewritten in terms of
\(\tilde{\rv D}^{+}\) and \(\tilde{\rv D}^{-}\) as follows.  
Since
\(\|\Rv D\|_1=\tilde{\rv D}^{+}+\tilde{\rv D}^{-}\) and
\(\lvert\tilde{\rv D}\rvert=\lvert\tilde{\rv D}^{+}-\tilde{\rv D}^{-}\rvert\), then
\[
  \|\Rv D\|_1-\lvert\tilde{\rv D}\rvert
       = 2\min\{\tilde{\rv D}^{+},\tilde{\rv D}^{-}\}
       = 2\tilde{\rv Z}.
\]
Hence,
\[
  \mathcal{E}_2=\{\tilde{\rv Z}>\lambda(\tilde{\rv D})\}.
\]
Taking complements gives
\[
  \bar{\mathcal{E}}_2
     =\bigl\{\tilde{\rv Z}\le\lambda(\tilde{\rv D})\bigr\}
     =\bigcup_{\Delta=-\ell N'}^{\ell N'}
        \bigcup_{z=0}^{\lambda(\Delta)}
           \{\tilde{\rv Z}=z,\tilde{\rv D}=\Delta\}.
\]
Since the events in the last union are mutually exclusive, we have
\[
  \Pr(\bar{\mathcal{E}}_2)
      =\sum_{\Delta=-\ell N'}^{\ell N'}
         \sum_{z=0}^{\lambda(\Delta)}
            \Pr(\tilde{\rv Z}=z,\tilde{\rv D}=\Delta),
\]
and subtracting from~\(1\) gives the result in Lemma~\ref{probE2}.
\end{proof}
Evaluating the probability of $\mathcal{E}_2$ based on Lemma~\ref{probE2} requires the joint PMF of $(\tilde{\rv{Z}}, \tilde{\rv{D}})$, which we derive next. The segment-level offsets $(\rv{D}_i)_{i=1}^{N'}$ are i.i.d. with support $[-\ell,\ell]$. The PMF of a generic segment-level offset $\rv{D}\in [-\ell,\ell]$ is given by
\begin{multline} \label{eqDd}
\Pr(\rv{D}=\delta) =  \sum_{j=\max\{0,-\delta\}}^{\left\lfloor (\ell-\delta)/2 \right\rfloor} \binom{\ell}{j, j+\delta, \ell - 2j - \delta} \\
\times (P_d)^j (P_i)^{j+\delta} (1-P_d-P_i)^{\ell-2j-\delta},
\end{multline}
which follows from the definition $\mathsf{D} = \mathsf{N}^{\mathsf{ins}} - \mathsf{N}^{\mathsf{del}}$ and the fact that $(\mathsf{N}^{\mathsf{del}}, \mathsf{N}^{\mathsf{ins}}, \ell-\mathsf{N}^{\mathsf{del}}- \mathsf{N}^{\mathsf{ins}})$ has a multinomial distribution with parameters \mbox{$(\ell\,;\,P_d,P_i,1-P_d-P_i)$}. 

Define the PMFs over the non-negative and non-positive supports of $\rv{D}$ by
\begin{align*}
\p&=(p_0,\ldots,p_\ell), \quad p_\delta \triangleq \Pr(\rv{D}=\delta),\,\delta\in [0,\ell], \\
\q&=(q_0,\ldots,q_\ell), \quad q_\delta \triangleq \Pr(\rv{D}=-\delta),\,\delta\in [0,\ell].
\end{align*}
Also, define the PMFs over the positive and negative supports (excluding zero) as
\[
        \bar{\p}\triangleq(\bar{p}_1,\ldots,\bar{p}_\ell),\qquad
        \bar{\q}\triangleq(\bar{q}_1,\ldots,\bar{q}_\ell),
\]
with $\bar{p}_{\delta}=p_{\delta}$ and $\bar{q}_{\delta}=q_{\delta}$ for all $\delta \in [1,\ell]$. For a vector $\x$, we write $\x^{(*j)}$ for its $j$-fold discrete convolution, and $x^{(*j)}_z$ as the $z$-th element of $\x^{(*j)}$. Claim~\ref{jointE2} provides the exact joint distribution of $(\tilde{\rv{Z}}, \tilde{\rv{D}})$. The proof of this claim is given in Appendix~\ref{appB}.
\begin{claim} \label{jointE2}
For $z\in [0,\ell N']$ and \mbox{$\Delta\in[-\ell N', \ell N']$}, the joint PMF of the random variables $(\tilde{\rv{Z}},\tilde{\rv{D}})$ is given by
\begin{equation*}
  \Pr(\tilde{\rv{Z}}=z,\tilde{\rv{D}}=\Delta)%
    = 
    \begin{cases}
        \displaystyle\sum_{j=j_{\text{min}}}^{j_{\text{max}}}
        \binom{N'}{j}\;
        \bar{q}^{(*j)}_{z}\;
        p^{(*\,N'-j)}_{\Delta+z},
  \;\; \Delta\ge 0, \\
     \displaystyle\sum_{j=j_{\text{min}}}^{j_{\text{max}}}
       \binom{N'}{j}\;
        \bar{p}^{(*j)}_{z}\;
        q^{(*\,N'-j)}_{\Delta+z},
  \;\; \Delta< 0,
    \end{cases}
\end{equation*}
where
\[
j_{\text{min}} = \left \lceil \frac{z}{\ell} \right \rceil, \quad
j_{\text{max}} = 
\min \left\{
      z,\;
      N' - \left\lceil \dfrac{|\Delta|+z}{\ell} \right\rceil
    \right\}.
\]
\end{claim}

\subsubsection{Probability of event $\mathcal{E}_3$}\label{err3}
The probability of $\mathcal{E}_3$ depends on the parity encoding function being used. One of the options discussed in Section~\ref{discF} is to encode the $\tilde{k} = c_2\ell$ check parity bits using an $(\tilde{n}, \tilde{k})$ code with minimum SLD $d^{\text{SLD}}_{\text{min}}$. In this case, the decoder can correct up to $\tilde{t} = \left\lfloor (d^{\text{SLD}}_{\text{min}} - 1)/2 \right\rfloor$ edit errors in the check parities with zero error by applying minimum distance decoding to the last $\tilde{n}$ bits of the channel output~$\y$. Therefore, an immediate upper bound on $\Pr(\mathcal{E}_3)$ under the SLD-based $(\tilde{n}, \tilde{k})$ code is
\begin{equation}
 \Pr(\mathcal{E}_3) \leq 1 - \sum_{j=0}^{\tilde{t}} \binom{\tilde{n}}{j} P_{\text{edit}}^j (1-P_{\text{edit}})^{\tilde{n}-j}.
\end{equation}
Another option discussed in Section~\ref{discF} is to use a $t$-fold repetition code and apply majority voting over non-overlapping sliding windows of size $t$ within the last $tc_2\ell$ bits of $\y$. Deriving a closed-form upper bound on $\Pr(\mathcal{E}_3)$ in this case is challenging, as the windows may be misaligned due to offsets and may include ``foreign'' bits not originating from the check parities. Nevertheless, we provide a theoretical analysis in Appendix~\ref{appC} that allows upper bounding this probability via a recursive dynamic programming approach, and we apply it in Section~\ref{simul}.

\subsection{Simulation Results} \label{simul}
We simulated the overall probability of decoding error (decoding failures + miscorrections) of the binary GC+ code under the channel model described in Section~\ref{model}. Our evaluation focused on short blocklengths (a few hundred) and moderate code rates (near or above 0.5). The complete Python source code is available at: \mbox{\texttt{\url{https://github.com/sergekashanna/GCPlus}}}.
To the best of our knowledge, existing binary edit-correcting codes in the literature are either impractical or not designed for the considered blocklength and rate regimes, and are therefore not suitable for comparison.

Specifically, we set the message length to $k = 140$ bits with the segmentation parameter $\ell = \lfloor \log k \rfloor = 7$. For i.i.d. edits, we fixed the channel parameter to $w = n$ and varied the edit probability $P_{\text{edit}}$, which corresponds to an average edit rate of $\varepsilon_{av} = P_{\text{edit}}$. For localized edits, we varied the channel parameter $w < n$ while keeping $P_{\text{edit}}$ fixed, in which case the average edit rate becomes $\varepsilon_{av} = P_{\text{edit}} \times \tfrac{w}{n}$. In decoding, the {\em general check} (Section~\ref{dec}) was applied for the i.i.d. case, whereas the {\em burst check} was used for localized edits.

\subsubsection{I.I.D. edits} We studied values of $P_{\text{edit}}$ ranging from $0.1\%$ to $1.5\%$, with both symmetric error probabilities ($P_s = P_d = P_i = P_{\text{edit}}/3$) and asymmetric ones  ($P_s = 0.53P_{\text{edit}}, P_d = 0.45P_{\text{edit}}, P_i = 0.02P_{\text{edit}}$). The chosen range of $P_{\text{edit}}$ and the asymmetric proportions are both motivated by recent experimental results on DNA data storage~\cite{gimpel2023digital}. The number of parities of the GC+ code was set to $(c_1, c_2) = (8,1)$, i.e., $c_1 = 8$ guess parities and $c_2 = 1$ check parity.

For protecting the $7$ bits corresponding to the binary representation of the $c_2=1$ check parity, we tested different encoding functions (Section~\ref{discF}): (i) $(7t, 7)$ repetition codes with $t=3$ and $t=5$, resulting in $(217,140)$ and $(231,140)$ GC+ codes with rates $R=0.645$ and $R=0.606$, respectively, and (ii) a $(20,7)$ code with $d^{\text{SLD}}_{\text{min}}=5$ constructed via exhaustive search, yielding a $(216,140)$ GC+ code with rate $R=0.648$. 

The results in Fig.~\ref{fig2} compare the theoretical evaluation from~(12) with the empirical simulations, showing good agreement in all cases. These results indicate that the GC+ code can effectively correct i.i.d. edits within the regime of interest while maintaining moderately high code rates ($R>0.6$) at short blocklengths. Among the three configurations, the setting where the check parity is protected using the $(20,7)$ code (Fig.~\ref{subc}) consistently achieved the lowest probability of decoding error, with average relative reductions of $6\%$ and $37\%$ compared to the repetition settings with $t=5$ (Fig.~\ref{subb}) and $t=3$ (Fig.~\ref{suba}), respectively. This configuration also offered the highest code rate among the three.  The gains in reliability and rate came at the expense of increased time for decoding the check parity, since the setup in Fig.~\ref{subc} relies on minimum-distance decoding via lookup tables, whereas the repetition codes are decoded using much faster sliding-window majority decoding. Nevertheless, this increase is amortized when the improvement in decoding error probability is pronounced, because in error events the decoder typically exhausts all predefined patterns, while in successful cases it usually terminates early, making the time cost of errors dominate the parity decoding cost in the average-case. Finally, the results show that decoding performance also depends on the proportions of $P_d, P_i,$ and $P_s$ relative to $P_{\text{edit}}$, with a higher probability of decoding error observed in the symmetric case.


\subsubsection{Localized edits} We studied the case $w<n$, with $w$ varying between $\ell+1=8$ and $4\ell+1=29$, while fixing \mbox{$P_{\text{edit}}=99\%$} and \mbox{$P_s=P_d=P_i=33\%$}. The number of parities was set to \mbox{$c_1=c_2=(w-1)/\ell+1$}, and the parity encoding function based on the buffer $\b$ of length $3(w+1)$ defined in Section~\ref{discF} was used. The resulting code rate $R$ varies with $w$. The results presented in~Fig.~\ref{fig3} demonstrate that for localized edits, the GC+ code can handle significantly higher edit rates $\varepsilon_{av}$ compared to the i.i.d. case, while achieving lower decoding error probabilities. Furthermore, using the {\em burst check} (Section~\ref{dec}) instead of the {\em general check} in this setting led to much faster decoding.

\begin{figure}[h!]
\centering
 \begin{tabular}{|c|c|c|c|}
\hline
Window length & Avg. edit rate & Code rate & Prob. error \\ \hline
$w=8$ & $4.1\%$ & $0.72$ & $=2.83e^{-4}$ \\ \hline
$w=15$ & $6.5\%$ & $0.61$ & $=1.00e^{-6}$ \\ \hline
$w=22$ & $8.2\%$ & $0.53$ & $<1.00e^{-6}$ \\ \hline
$w=29$ & $9.6\%$ & $0.47$ & $<1.00e^{-6}$ \\ \hline
\end{tabular}
\caption{
Empirical probability of decoding error of the binary GC+ code under localized edits for varying window lengths $w$ (channel parameter). The edit probability is fixed at $P_{\text{edit}}=99\%$ with $P_s=P_d=P_i$. Code parameters are $k=140$, $\ell=\lfloor \log k \rfloor$, and $c_1=c_2=(w-1)/\ell+1$. For each $w$, the corresponding average edit rate $\varepsilon_{av}$ and code rate $R$ are reported. Buffer-based parity encoding is applied. Results are averaged over $10^6$ independent runs, with entries $<1.0e^{-6}$ indicating no decoding errors observed.}
\label{fig3}
\end{figure}

\section{Application to DNA Data Storage} \label{DNA}
In this section, we demonstrate how the GC+ code can be applied in the context of DNA storage to enhance data reliability. The general workflow for encoding binary data into DNA is illustrated in Fig.~\ref{fig1}. Due to the limitations of current DNA synthesis technologies, binary data is encoded into a collection of short DNA sequences, called oligos, typically a few hundred nucleotides in length. This is accomplished by partitioning the data into short, non-overlapping fragments and mapping each binary fragment into a sequence of nucleotides~$\{\mathtt{A,C,G,T}\}$. In practice, metadata such as oligo indices is also added to facilitate the reconstruction of the original data during retrieval. 

Within this workflow, error correction can be introduced at two levels to mitigate errors arising from DNA synthesis, storage, and sequencing. First, an \emph{outer code} is applied across the full set of oligos to generate redundant oligos. Then, an \emph{inner code} is applied to each oligo individually by adding redundancy within it. The two codes serve complementary purposes. The inner code is responsible for correcting errors within an individual oligo, which places strict demands on its design: it must operate at short blocklengths and, ideally, handle all three types of edit errors (deletions, insertions, and substitutions). The outer code is intended to recover lost oligos and/or correct residual errors from the inner code. In coding-theoretic terms, the role of the outer code is to correct erasures and substitutions, which are easier to handle than deletions and insertions addressed by the inner code. Moreover, since the number of oligos (fragments) depends on the data size and is typically large, code length is not a limiting factor for the outer code, making standard codes such as RS, LDPC, or Raptor well suited for this role.

\begin{figure}[h!]
\begin{center}
\resizebox{0.4\textwidth}{!}{
\begin{tikzpicture}[>=stealth, node distance=2.5cm, auto]
\tikzstyle{block} = [rectangle, draw, fill=lightgray!20, text centered, rounded corners, minimum height=2.5em]
\tikzstyle{line} = [draw, -latex']

\node (file) at (-1,0) {\includegraphics[width=1cm]{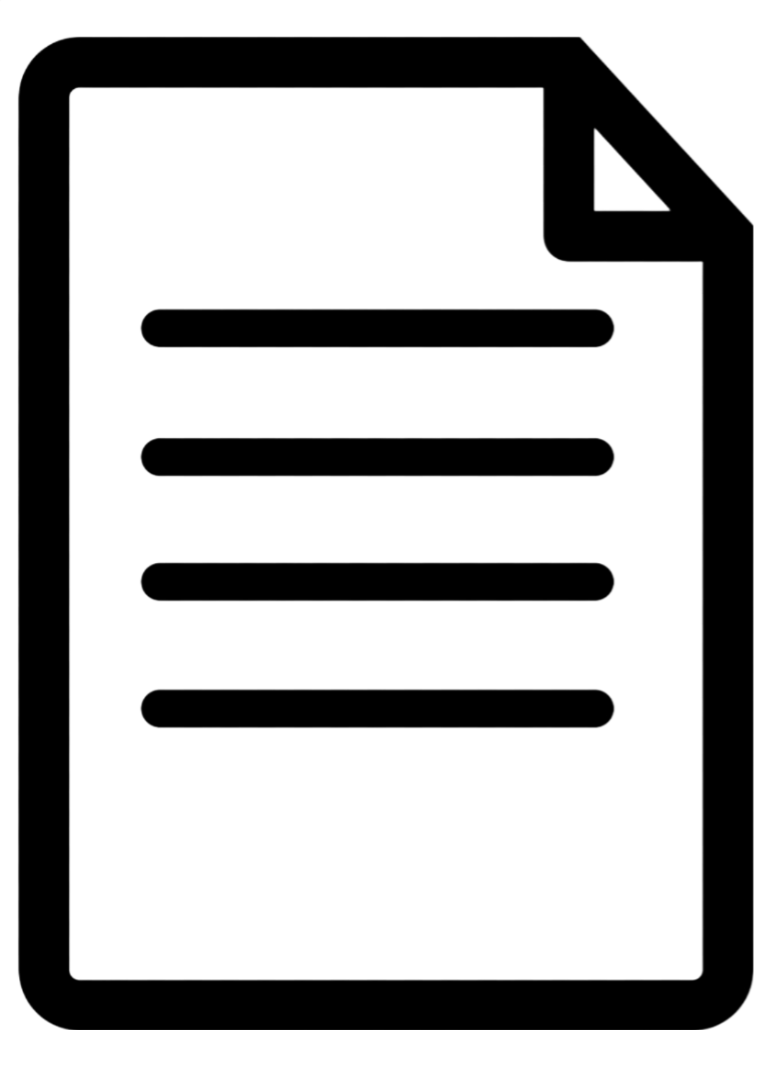}};
\node [font=\small, text width=2cm, align=center] at (-1,-1.05cm) {Compressed Binary Data};
\node [block, right=1.65cm of file] (fragment) {Fragmentation};
\node [block, right=1.65cm of fragment] (outer) {Outer Code};
\node [block, below=1cm+0.5em of outer] (inner) {Inner Code};
\node [block, below=1cm+0.5em of fragment] (dna) {DNA Coding};
\draw[lightgray!380, line width=2pt] (-1.5, -1.6) -- (-0.5, -1.6);
\draw[lightgray!350, line width=2pt] (-1.5, -1.75) -- (-0.5, -1.75);
\node[] (storage) at (-0.5, -2.05) {};
\draw[lightgray!320, line width=2pt] (-1.5, -1.9) -- (-0.5, -1.9);
\draw[lightgray!270, line width=2pt] (-1.5, -2.05) -- (-0.5, -2.05);
\draw[lightgray!240, line width=2pt] (-1.5, -2.2) -- (-0.5, -2.2);
\draw[lightgray!210, line width=2pt] (-1.5, -2.35) -- (-0.5, -2.35);
\draw[lightgray!180, line width=2pt] (-1.5, -2.5) -- (-0.5, -2.5);
\node [font=\small, text width=2cm, align=center] at (-1,-2.8cm) {DNA Oligos};
\node [font=\small, text width=3cm, align=center] at (6.15,0.76cm) {Error-correction};

\path [line] (file) -- (fragment);
\path [line] (fragment) -- (outer);
\path [line] (outer) -- (inner);
\path [line] (inner) -- (dna);
\path [line] (dna) -- (storage);

\draw[dashed, rounded corners] (5,-2.61) rectangle (7.25,0.55);
\end{tikzpicture} }
\caption{Encoding binary data into DNA oligos with error-correction.}
\label{fig1}
\end{center}
\end{figure}

We propose integrating the GC+ code as the inner code in the workflow depicted in Fig.~\ref{fig1}. We argue that the GC+ code is a strong candidate for this role because: {\em (i)}~It can effectively correct deletions, insertions, and substitutions at short blocklengths, adhering to the length constraints dictated by DNA synthesis; {\em (ii)}~By improving the reliability of individual oligo recovery, it reduces the number of sequencing reads required for reconstruction, thereby lowering sequencing costs; {\em (iii)}~It addresses edit errors that arise during synthesis and storage, whereas the effectiveness of consensus-based strategies such as multiple sequence alignment is largely limited to sequencing errors~\cite{press2020hedges}; and {\em (iv)}~Its ability to detect and signal decoding failures allows the outer code to treat such cases as erasures rather than substitutions, which optimizes decoding. 

Next, we present simulation results on the performance of the GC+ code in the context of DNA storage. In these simulations, we focus on unconstrained DNA coding (2 bits/NT) given by the mapping:~$ 00 \leftrightarrow \mathtt{A}$, $01 \leftrightarrow \mathtt{C}$, $10 \leftrightarrow \mathtt{G}$, $11 \leftrightarrow \mathtt{T}$. We first study the performance of the GC+ code in isolation in the quaternary domain (Section~\ref{sec:gc_quaternary}), evaluating its decoding error probability both empirically and theoretically, and comparing it to the HEDGES code~\cite{press2020hedges}. We then turn to the complete pipeline (Fig.~\ref{fig1}) in Section~\ref{sec:gc_rs}, where GC+ is integrated as the inner code and an RS code is used as the outer code.


\subsection{Performance of GC+ in the Quaternary Domain} \label{sec:gc_quaternary}
We evaluated the performance of the GC+ code in the quaternary domain under the aforementioned unconstrained binary-to-quaternary mapping. Specifically, starting with a binary message, we first encoded the message using the GC+ code, and then mapped the resulting binary codeword into a quaternary DNA sequence. The random channel model described in Section~\ref{model} is then applied to the quaternary sequence, and the channel output is mapped back to binary for GC+ decoding. Under this setup, each nucleotide edit in the quaternary domain corresponds to up to two consecutive bit edits in the underlying binary code. Moreover, if the segmentation parameter $\ell$ of the GC+ code is chosen to be an {\em even} integer, the bit edits induced by a nucleotide edit remain confined to a single segment of the codeword. This simplifies decoding by establishing an equivalence between the binary and quaternary domains. Importantly, this equivalence also enables extending the theoretical results in Section~\ref{th:analysis} from binary to quaternary by simply evaluating the segment length at $\ell/2$ instead of $\ell$, while keeping the number of information segments $K=k/\ell$ and all other parameters unchanged.

We compared the performance of the GC+ code in the quaternary domain with the HEDGES code~\cite{press2020hedges}, using the Python implementation provided in~\cite{HEDGES}. HEDGES is a convolutional code that directly encodes a binary bitstream into quaternary DNA symbols in a single step. It is practical for short code lengths, and its decoding algorithm is specifically designed to correct edits within a single DNA sequence, in contrast to recent decoding implementations of convolutional codes designed for decoding over multiple reads of the same codeword~\cite{trellisBMA,maarouf2022concatenated}. We evaluated the decoding error probability of both codes under i.i.d. and localized edits (as in Section~\ref{simul}) by applying the channel model over the quaternary alphabet.

In our simulations, we considered binary messages of length $168$ bits encoded into DNA sequences, where the amount of information bits encoded per DNA nucleotide defines the {\em information density} in bits/NT denoted by $\rho$ (higher is better). The codes are parametrized as follows:
\begin{itemize}[leftmargin=*]
\item {\em HEDGES: }  A code of length $176$~NTs with information density $\rho \approx 0.95$ bits/NT, generated by encoding a binary message of length $168$ bits using a convolutional code of rate~$0.5$ and appending one runout byte. The hyperparameters of HEDGES are set to the standard values from~\cite{press2020hedges}:  greediness parameter $p_{ok}=-0.1$, heap size $10^6$, a single runout byte, and unconstrained output.
\item {\em GC+: } For i.i.d. edits, we set $\ell=8$, $(c_1,c_2)=(8,1)$, $\lambda=1$ for $|\Delta|\leq 1$ and $\lambda=0$ otherwise. To protect the $4$ NTs (8 bits) corresponding to the $c_2=1$ check parity, we used a quaternary $(12,4)$ code with $d^{\text{SLD}}_{\text{min}}=5$ constructed via exhaustive search. A binary message of length $168$ bits is thus encoded into a DNA sequence of length $128$ NTs, yielding $\rho \approx 1.31$ bits/NT in the i.i.d. setting. For localized edits, we set $(c_1,c_2)=(2,2)$, $\ell=8$, and use buffer-based parity encoding, with code lengths and rates varying according to the channel parameter $w$.
\end{itemize}

The results for the i.i.d. case presented in Fig.~\ref{fig4} show that the GC+ code attains a lower decoding error probability than HEDGES for $P_{\text{edit}}$ between $0.1\%$ and $1.5\%$, while achieving a much higher information density of $1.31$ bits/NT compared to $0.95$ bits/NT for HEDGES. Furthermore, as in the binary case (Fig.~\ref{fig2}), the empirical and theoretical performance of the GC+ code are in close agreement, and the decoding error probability is lower for the asymmetric error proportions~($P_s = 0.53P_{\text{edit}}, P_d = 0.45P_{\text{edit}}, P_i = 0.02P_{\text{edit}}$) than for the symmetric ones ($P_s = P_d = P_i = P_{\text{edit}}/3$). Interestingly, this behavior contrasts with that of HEDGES, which performs worse in the asymmetric case.

The results for localized edits shown in Fig.~\ref{fig5} show that the HEDGES code is highly sensitive to the presence of consecutive edits, a behavior common among convolutional codes. In contrast, the GC+ code exhibits significantly better decoding performance in such scenarios. Note that the simulated GC+ codes in Fig.~\ref{fig5} are shorter and have higher information densities than the HEDGES code. Consequently, for a given value of $w$, the GC+ codes experience higher average edit rates $\varepsilon_{av}=P_{\text{edit}} \times \frac{w}{n}$ compared to HEDGES.

\begin{figure}[h!]
\begin{center}
\begin{tikzpicture}[scale=0.85]
\begin{semilogyaxis}[
    legend cell align={left},
    tick scale binop=\times,
    xlabel={Edit Probability $P_{\text{edit}}$ },
    ylabel={Probability of Decoding Error},
    grid=both,
    minor tick num=1,
    xtick={0.001, 0.003, 0.005, 0.007, 0.009, 0.011, 0.013, 0.015},
    ytick={1e-6, 1e-5, 0.0001, 0.001, 0.01, 0.1},
    legend pos=south east,
    legend style={font=\scriptsize},
    xmajorgrids=true,
    xminorgrids=true,
    ymajorgrids=true,
    yminorgrids=true,
    major grid style={black!15}, 
    minor grid style={dotted,black!40}, 
    ymax=0.11,
    ymin=1e-6,
    xmin=0.001,
    xmax=0.015,
]

\addplot+[color=blue,mark=o,mark size=1.5, mark options={fill=none},line width=0.6pt] coordinates {
(0.001,1.0000e-06)
(0.002,1.1000e-05)
(0.003,4.6000e-05)
(0.004,1.3200e-04)
(0.005,3.3800e-04)
(0.006,6.7800e-04)
(0.007,1.2230e-03)
(0.008,1.9640e-03)
(0.009,3.0730e-03)
(0.010,4.4950e-03)
(0.011,6.4380e-03)
(0.012,8.7890e-03)
(0.013,1.1549e-02)
(0.014,1.4783e-02)
(0.015,1.8952e-02)
%
%
};

\addplot+[color=blue,mark=none,mark size=1.5, mark options={solid, fill=none,rotate=180},line width=0.8pt,dashed] coordinates {
(0.001,1.28972e-06)
(0.002,1.80028e-05)
(0.003,8.40335e-05)
(0.004,0.000248648)
(0.005,0.000572133)
(0.006,0.00112227)
(0.007,0.00197149)
(0.008,0.00319457)
(0.009,0.00486677)
(0.010,0.00706235)
(0.011,0.00985338)
(0.012,0.0133088)
(0.013,0.0174935)
(0.014,0.0224681)
(0.015,0.028288)
};

\addplot+[color=red,mark=o,mark size=1.5,mark options={fill=none},line width=0.6pt] coordinates {
(0.001,0.000000)
(0.002,0.0000015)
(0.003,9.000e-06)
(0.004,0.0000285)
(0.005,0.0000700)
(0.006,0.0001600)
(0.007,0.0003160)
(0.008,0.0005660)
(0.009,0.0009295)
(0.010,0.0014605)
(0.011,0.0022225)
(0.012,0.0032965)
(0.013,0.0047205)
(0.014,0.0063515)
(0.015,8.337e-03)
};

\addplot+[color=red,mark=none,mark size=1.5, mark options={solid, fill=none,rotate=180},line width=0.8pt,dashed] coordinates {
(0.001,2.96874e-07)
(0.002,3.3729e-06)
(0.003,1.58171e-05)
(0.004,5.00433e-05)
(0.005,0.000125395)
(0.006,0.000268593)
(0.007,0.00051366)
(0.008,0.000901447)
(0.009,0.00147887)
(0.010,0.00229792)
(0.011,0.00341459)
(0.012,0.00488768)
(0.013,0.00677764)
(0.014,0.00914539)
(0.015,0.0120512)
};

\addplot+[color=customgreen,mark=square,mark size=1.5, mark options={fill=none},line width=0.6pt] coordinates {
(0.001,0.00426)
(0.002,0.00900)
(0.003,0.01283)
(0.004,0.01784)
(0.005,0.02376)
(0.006,0.02799)
(0.007,0.03296)
(0.008,0.03894)
(0.009,0.04354)
(0.010,0.04829)
(0.011,0.05433)
(0.012,0.05980)
(0.013,0.06449)
(0.014,0.06991)
(0.015,0.07736)
};

\addplot+[color=custombrown,mark=square,mark size=1.5, mark options={fill=none},line width=0.6pt, solid] coordinates {
    (0.001,0.006)
    (0.002,0.01195)
    (0.003,0.01807)
    (0.004,0.02449)
    (0.005,0.02989)
    (0.006,0.03737)
    (0.007,0.04446)
    (0.008,0.05057)
    (0.009,0.05893)
    (0.01,0.06525)
    (0.011,0.07152)
    (0.012,0.07945)
    (0.013,0.08517)
    (0.014,0.09405)
    (0.015,0.10008)
};

\legend{}
\addlegendentry{GC+ Empirical (sym.)}
\addlegendentry{GC+ Theoretical (sym.)}
\addlegendentry{GC+ Empirical (asym.)}
\addlegendentry{GC+ Theoretical (asym.)}
\addlegendentry{HEDGES Empirical (sym.)}
\addlegendentry{HEDGES Empirical (asym.)}

\end{semilogyaxis}
\end{tikzpicture}
    \end{center}
        \caption{Probability of decoding error of the GC+ and HEDGES~\cite{press2020hedges} codes for a binary input message of length $168$ bits under i.i.d. edit errors. The GC+ and HEDGES codes have lengths $128$ and $176$~NTs, with information densities $\rho \approx 1.31$ and $\rho \approx 0.95$~bits/NT, respectively. Code parameters and edit error proportions (“sym.” and “asym.”) are specified in the text. Empirical results are averaged over $2\times 10^6$ runs for GC+ and $10^5$ runs for HEDGES.}
        

    \label{fig4}
    \vspace{-0.5cm}
\end{figure}

\begin{figure}[h!]
\begin{center}
\begin{tikzpicture}[scale=0.85]
\hspace{-0.2cm}
\begin{semilogyaxis}[
    legend cell align={left},
    tick scale binop=\times,
    xlabel={Window Length $w$ },
    ylabel={Probability of Decoding Error},
    grid=both,
    xtick={2,3,4,5},
    ytick={0.00001, 0.0001, 0.001, 0.01, 0.1, 1},
    legend style={at={(0.7,0.085)},anchor=west,font=\scriptsize},
    xmajorgrids=true,
    xminorgrids=false,
    ymajorgrids=true,
    yminorgrids=true,
    major grid style={black!10}, 
    minor grid style={dotted,black!15}, 
    ymax=1,
    ymin=0.00001,
    xmin=2,
    xmax=5,
]

\addplot+[mark=triangle,mark size=1.5, mark options={fill=none},line width=0.6pt] coordinates {
    (2,5.80E-05)
    (3,6.40E-05)
    (4,7.10E-05)
    (5,8.30E-05)
};

\addplot+[mark=o,mark size=1.5, mark options={fill=none},line width=0.6pt] coordinates {
    (2,1.40E-01)
    (3,2.30E-01)
    (4,3.70E-01)
    (5,5.70E-01)
};

\node at (4.23cm,-2.44cm) {
\resizebox{0.27\textwidth}{!}{%
         \begin{tabular}{|c|c|c|c|c|c|}
\hline
\multirow{2}{*}{$w$} & \multicolumn{3}{c|}{GC+} & \multicolumn{2}{c|}{HEDGES} \\ \cline{2-6} 
 & Length &  $\varepsilon_{av}$ & $\rho$ & $\varepsilon_{av}$ & $\rho$ \\ \hline
 $2$ & $108$ & $1.8\%$ & $1.56$ & $1.1\%$ & $0.95$ \\ \hline
 $3$ & $111$ & $2.7\%$ & $1.51$ & $1.7\%$ & $0.95$ \\ \hline
 $4$ & $114$ & $3.5\%$ & $1.47$ & $2.3\%$ & $0.95$ \\ \hline
 $5$ & $117$ & $4.2\%$ & $1.44$ & $2.8\%$ & $0.95$ \\ \hline
\end{tabular} }
    };

\legend{}
 \addlegendentry{GC+}
 \addlegendentry{HEDGES}
\end{semilogyaxis}
\end{tikzpicture}
    \end{center}
    \caption{Empirical probability of decoding error of the GC+ and HEDGES~\cite{press2020hedges} codes for a binary input message of length $168$ bits under localized edit errors with varying window length $w$ (channel parameter). The edit probability is fixed at $P_{\text{edit}}=99\%$ with $P_s=P_d=P_i=33\%$. The same HEDGES code of length 176 NTs as in Fig.~\ref{fig4} is evaluated. The GC+ code parameters are specified in the text. Reported values include information densities $\rho$ (bits/NT), average edit rates $\varepsilon_{av}$, and the GC+ code lengths (in NTs) for different $w$. Results are averaged over $10^5$ runs for GC+ and $10^3$ runs for HEDGES.}
    \label{fig5}
    \vspace{-0.5cm}
\end{figure}

\subsection{GC+ as Inner Code in DNA Storage} \label{sec:gc_rs}
Following the workflow in Fig.~\ref{fig1}, we consider synthetic binary data of size $1.68$ Mb as input, where each bit is generated independently according to the $\text{Bernoulli}(0.5)$ distribution. The data is partitioned into $10^4$ non-overlapping fragments of length $168$ bits. These fragments are encoded using an outer systematic RS code over $\F_{2^{14}}$ of rate $R_{\text{out}}$, producing \mbox{$10^4(R_{\text{out}}^{-1}-1)$} additional ``parity'' fragments of the same size. Each fragment is then encoded with an inner GC+ code of rate $R_{\text{in}}$. The resulting binary GC+ codewords are mapped into DNA using the 2 bits/NT binary-to-quaternary conversion, yielding $10^4R_{\text{out}}^{-1}$ oligos, each of length $84R_{\text{in}}^{-1}$ NTs. 

We apply the same channel model as in previous sections, with i.i.d. edits over the quaternary alphabet and error proportions motivated by DNA storage\cite{gimpel2023digital} ($P_s = 0.53P_{\text{edit}}, P_d = 0.45P_{\text{edit}}, P_i = 0.02P_{\text{edit}}$). To assess error-correction performance in isolation, we assume that each oligo is read exactly once and that reads are ordered, thereby omitting other retrieval modules such as clustering and reconstruction (e.g., alignment and consensus), which are typically required in practice when handling multiple unordered reads at the sequencer output.

Decoding proceeds in two stages. Each oligo is first decoded using the inner GC+ code. If decoding fails, the oligo is flagged and treated as an erasure by the outer RS code. Otherwise, the output of the inner decoder is passed directly to the outer RS code, possibly containing substitution errors at the RS-symbol level in $\F_{2^{14}}$ due to GC+ miscorrections. The outer RS code is then responsible for correcting both erasures and substitutions across all fragments. Data retrieval succeeds if these erasures and substitutions fall within the error-correction capability of the RS code.

With this setup in place, we analyzed the optimal trade-offs between the rates of the inner and outer codes by evaluating the maximum achievable information density $\rho=2R_{\text{in}}R_{\text{out}}$ at which data retrieval still succeeded with 100\% success rate over ten independently and randomly generated datasets of size 1.68~Mb. Specifically, we considered several configurations of the inner GC+ code given in Fig.~\ref{figa1} and, for each one, performed a grid search over the feasible outer RS rates to determine $R^*_{\text{out}}$ that maximizes the information density $\rho^*=2R_{\text{in}}R^*_{\text{out}}$ for every value of $P_{\text{edit}}$. 

The results are summarized in Fig.~\ref{figa}. Fig.~\ref{figa3} plots the maximum achievable information densities for the inner code configurations in Fig.~\ref{figa1} as a function of $P_{\text{edit}}$. The ``uncoded'' case corresponds to $R_{\text{in}}=1$, where oligos of incorrect length are flagged as erasures and only correctly sized oligos are passed to the RS code. As shown in Fig.~\ref{figa4}, in the low-error regime ($P_{\text{edit}}\leq 0.6\%$), the highest information density is attained without inner coding, i.e., by allocating all redundancy to the outer RS code. For higher edit rates, however, GC+ inner codes enable strictly better information densities, with the codes achieving the best densities among those considered being: code~$\mathrm{C}$ for $P_{\text{edit}}\in\{0.7\%,0.8\%,0.9\%,1.0\%\}$, and code~$\mathrm{D}$ for $P_{\text{edit}}\in\{1.1\%,1.2\%,1.3\%,\ldots,1.5\%\}$. These results demonstrate that GC+ inner coding enables reliable data retrieval at higher information densities when edit errors are significant, which directly translates into lower synthesis cost per information bit.

\begin{figure}[h!]
    \centering
        \subfloat[Inner GC+ code configurations.]{
        \centering
        \resizebox{0.4\textwidth}{!}{
            \begin{tabular}{ccccc}
            \hline
            GC+ Code & Rate $R_{\text{in}}$ & Oligo length&  $c_1$ & $c_2$ \\ \hline
            $\mathrm{A}$ & $0.81$ & $104$ & $2$ & $1$ \\ 
            $\mathrm{B}$ & $0.78$ & $108$ & $3$ & $1$ \\ 
            $\mathrm{C}$ & $0.75$ & $112$ & $4$ & $1$ \\ 
            $\mathrm{D}$ & $0.72$ & $116$ & $5$ & $1$ \\ 
            $\mathrm{E}$ & $0.70$ & $120$ & $6$ & $1$ \\ 
            $\mathrm{F}$ & $0.68$ & $124$ & $7$ & $1$ \\  \hline
            \end{tabular} } 
        \label{figa1} } \\ \vspace{0.4cm}
            \subfloat[Optimal information densities resulting from inner + outer codes.]{
            \centering
            \resizebox{0.45\textwidth}{!}{
            \begin{tikzpicture}[scale=1]
    \begin{axis}[
        legend cell align={left},
        tick scale binop=\times,
        xlabel={Edit Probability $P_{\text{edit}}$ },
        ylabel={Information Density $\rho$ (bits/NT)},
        cycle list={
                {red},
                {customgreen},
                {purple},
                {blue},
                {orange},
                {custombrown},
                {black},
                {yellow}
            },
                            grid=both,
        xtick={0.005,  0.006, 0.007, 0.008, 0.009, 0.010, 0.011, 0.012, 0.013, 0.014, 0.015},
       ytick={0.9,1,1.1,1.2,1.3,1.4,1.5, 1.6},
        legend pos=south west,
        legend style={font=\scriptsize},
        xmajorgrids=true,
        xminorgrids=true,
        ymajorgrids=true,
        yminorgrids=true,
    minor xtick={0.006,0.008,0.010,0.012,0.014}, 
    minor ytick={0.925,0.95,0.975,1.025,1.05,1.075,1.125,1.15,1.175,1.225, 1.250, 1.275, 1.325, 1.350, 1.375, 1.425, 1.450, 1.475, 1.525, 1.550, 1.575},
        major grid style={black!15}, 
        minor grid style={dotted,black!40}, 
        ymax=1.6,
        ymin=0.98,
                           xmin=0.005,
   		 xmax=0.015,
    ]

    \addplot+[mark=o,mark options={fill=none},mark size=2,line width=0.5pt] coordinates {		
(0.005,1.4134) 
(0.006,1.3487) 
(0.007,1.2731) 
(0.008,1.1953)
(0.009,1.1120) 
(0.010,1.0257)
    };
    
    \addplot+[mark=square,mark options={fill=none},mark size=2,line width=0.5pt] coordinates {    				
(0.005,1.4543) (0.006,1.4154) (0.007,1.3766) (0.008,1.3300)
(0.009,1.2755) (0.010,1.2155) (0.011,1.1511) (0.012,1.0811)
(0.013,1.0111)
    };
    
    \addplot+[mark=pentagon,mark options={fill=none},mark size=2,line width=0.5pt] coordinates { 
(0.005,1.4787) (0.006,1.4683) (0.007,1.4549) (0.008,1.4395)
(0.009,1.4207) (0.010,1.3999) (0.011,1.3757) (0.012,1.3499)
(0.013,1.3212) (0.014,1.2950) (0.015,1.2742)
    };
    
    \addplot+[mark=diamond,mark options={fill=none},mark size=2,line width=0.5pt] coordinates { 
(0.005,1.4356) (0.006,1.4294) (0.007,1.4215) (0.008,1.4102)
(0.009,1.3967) (0.010,1.3806) (0.011,1.3636) (0.012,1.3435)
(0.013,1.3203) (0.014,1.2966) (0.015,1.2727)
    };
    
    \addplot+[mark=triangle,mark options={fill=none},mark size=2,line width=0.5pt] coordinates {
(0.005,1.3977) (0.006,1.3963) (0.007,1.3940) (0.008,1.3911)
(0.009,1.3872) (0.010,1.3827) (0.011,1.3772) (0.012,1.3710)
(0.013,1.3632) (0.014,1.3548) (0.015,1.3476)
    };
    
        \addplot+[mark=mystar,mark options={fill=none},mark size=2,line width=0.5pt] coordinates {
(0.005,1.3532) (0.006,1.3528) (0.007,1.3517) (0.008,1.3504)
(0.009,1.3484) (0.010,1.3460) (0.011,1.3426) (0.012,1.3387)
(0.013,1.3341) (0.014,1.3299) (0.015,1.3263)
    };
    
       \addplot+[mark=10-pointed star,mark size=2,line width=0.5pt, densely dashdotted] coordinates {	  
(0.005,1.5679)
(0.006,1.4920)
(0.007,1.4220)
(0.008,1.3540)
(0.009,1.2920)
};
    
    \legend{}
    \addlegendentry{Code $\mathrm{A}$}
     \addlegendentry{Code $\mathrm{B}$}
     \addlegendentry{Code $\mathrm{C}$}
     \addlegendentry{Code $\mathrm{D}$}
     \addlegendentry{Code $\mathrm{E}$}
     \addlegendentry{Code $\mathrm{F}$}
         \addlegendentry{Uncoded}
    \end{axis}
    \end{tikzpicture} }
        \label{figa3} } \\ \vspace{0.4cm}
    \subfloat[Inner codes yielding highest information densities.]{
            \centering
            \resizebox{0.45\textwidth}{!}{
\begin{tikzpicture}[scale=1]
\begin{axis}[
    legend cell align={left},
    tick scale binop=\times,
    xlabel={Edit Probability $P_{\text{edit}}$ },
    ylabel={Information Density $\rho$ (bits/NT)},
    cycle list={
            {red},
            {blue},
            {orange},
            {custombrown},
            {pink},
            {gray}
        },
    grid=both,
    xtick={0.005,  0.006, 0.007, 0.008, 0.009, 0.010, 0.011, 0.012, 0.013, 0.014, 0.015},
    ytick={1.30,1.35, 1.40,1.45, 1.50, 1.55,1.60},
    legend pos=south west,
        minor xtick={0.006,0.008,0.010,0.012,0.014}, 
    minor ytick={1.325, 1.350, 1.375, 1.425, 1.450, 1.475, 1.525, 1.550, 1.575},
    legend style={font=\scriptsize},
    xmajorgrids=true,
    xminorgrids=true,
    ymajorgrids=true,
    yminorgrids=true,
    major grid style={black!15}, 
    minor grid style={dotted,black!40}, 
    ymax=1.6,
    ymin = 1.28,
                       xmin=0.005,
   		 xmax=0.015,
]

\addplot+[only marks,mark=10-pointed star,mark size=2, color=black] coordinates {
(0.005,1.5679)
(0.006,1.4920)
};

\addplot+[only marks, mark=pentagon*,mark size=2, color=purple] coordinates {
(0.007,1.4549) (0.008,1.4395)
(0.009,1.4207) (0.010,1.3999)
};

\addplot+[only marks, mark=triangle*,mark size=2] coordinates {
(0.011,1.3772) (0.012,1.3710)
(0.013,1.3632) (0.014,1.3548) (0.015,1.3476)
};
%

\legend{}
\addlegendentry{Uncoded}
\addlegendentry{Code $\mathrm{C}$}
 \addlegendentry{Code $\mathrm{E}$}
\end{axis}
\end{tikzpicture} }
    \label{figa4} }
        \caption{Maximum achievable information densities $\rho$ obtained from combining the inner GC+ codes in (a) with an outer RS code. For each inner code configuration, the outer code rate that maximizes $\rho$ while ensuring 100\% data retrieval success over ten independently generated random datasets of size 1.68~Mb was determined using a grid search. Subfigure (b) shows the resulting densities as a function of $P_{\text{edit}}$, and (c) highlights the inner codes achieving the highest overall $\rho$ for each $P_{\text{edit}}$. For all GC+ codes in (a), the check parity is protected with a quaternary $(12,4)$ code, and the shared parameters are $\ell=8$, with $\lambda=1$ for $|\Delta|\leq 1$ and $\lambda=0$ otherwise.}
        \label{figa}
\end{figure}

\section{Conclusion}
In this work, we introduced the GC+ code, a systematic binary code designed to correct edit errors at short code lengths suitable for DNA storage applications. We evaluated its performance both theoretically and empirically. In the DNA storage context, we focused on a simplified setting with minimal read costs, assuming one read per oligo. Some interesting directions for future research include:
\begin{itemize}[leftmargin=*]
    \item Improving the decoding complexity of the GC+ code, particularly for the {\em general check} step. In~\cite{ma2021maximum}, the authors provided valuable insights on achieving fast decoding with the original version of Guess \& Check codes~\cite{GC} using maximum-likelihood estimation of deletion patterns on trellis graphs. It would be interesting to investigate similar approaches for edit correction.
    \item Considering other mathematical edit error models or noise simulators based on real data for {\em in silico} studies, in addition to conducting {\em in vitro} wet-lab experiments.
    \item Extending the results in Section~\ref{sec:gc_rs} by studying the achievable trade-offs between the rate of the inner GC+ code and the rate of the outer RS code in the presence of multiple reads per oligo.
    \item Designing quaternary edit correcting codes where both the encoding and decoding are performed over the DNA alphabet.
\end{itemize}




\section*{Acknowledgment} 
This work was supported by the French government through the France 2030 investment plan managed by the National Research Agency (ANR), as part of the Initiative of Excellence Université Côte d’Azur under reference number ANR-15-IDEX-01. 

\bibliographystyle{IEEEtran}
\bibliography{Refs}

\appendices
\section{Proof of Lemma~\ref{prop1}} \label{appA}
To prove Lemma~\ref{prop1}, we count the number of solution vectors \mbox{$\boldsymbol{\delta}=(\delta_1,\delta_2,\ldots,\delta_{N'})\in \mathbb{Z}^{N'}$} satisfying the equation 
\begin{equation} \label{aeq1}
    \delta_1+\delta_2+\ldots+\delta_{N'}=\Delta,
\end{equation}
subject to the constraints $\| \boldsymbol{\delta} \|_0 \leq c_1$, and \mbox{$\| \boldsymbol{\delta} \|_1 \leq \left \lvert \Delta \right \rvert + 2\lambda$}, where  $\Delta \in \mathbb{Z}$, $c_1 \in \mathbb{N}$, and $\lambda \in \mathbb{N}$. Due to the symmetry of the counting problem with respect to $\Delta$, we have $\left \lvert \mathcal{P}_2(\Delta,N',c_1,\lambda) \right \rvert = \left \lvert \mathcal{P}_2(-\Delta,N',c_1,\lambda) \right \rvert$. Hence, without loss of generality, we assume that $\Delta\geq 0$.

The condition $\| \boldsymbol{\delta} \|_0 \leq c_1$ restricts the number of non-zero elements in any solution vector $\boldsymbol{\delta}$ to a maximum of $c_1$. Let \mbox{$i_1\in \{0,1,\ldots,c_1\}$} be the number of non-zero elements in \mbox{$\boldsymbol{\delta}=(\delta_1,\delta_2,\ldots,\delta_{N'})$}. For a given $i_1$, there are $\binom{N'}{i_1}$ choices for selecting the indices of the non-zero elements.

The condition $\| \boldsymbol{\delta} \|_1 \leq \left \lvert \Delta \right \rvert + 2\lambda$, combined with \eqref{aeq1}, limits the sum of the negative elements in any solution vector~$\boldsymbol{\delta}$ to a minimum value of $-\lambda$. To prove this claim, assume for the sake of contradiction that there exists a solution vector $\boldsymbol{\delta}$ such that \mbox{$\sum_{i : \delta_i<0} \delta_i <-\lambda$}. Then, we have
\begin{align}
    \| \boldsymbol{\delta} \|_1  &= \sum_{i : \delta_i\geq0} \left \lvert \delta_i \right \rvert + \sum_{i : \delta_i<0} \left \lvert \delta_i \right \rvert, \label{aeq3}\\
    &= \sum_{i : \delta_i\geq0} \delta_i -\sum_{i : \delta_i<0} \delta_i \label{aeq4}, \\
    &= \Delta - 2 \sum_{i : \delta_i<0} \delta_i \label{aeq5}, \\
    &>  \Delta + 2\lambda, \label{aeq6}
\end{align}
where \eqref{aeq5} follows from \eqref{aeq1}, and \eqref{aeq6} follows from the assumption. The inequality in \eqref{aeq6} violates the condition \mbox{$\| \boldsymbol{\delta} \|_1 \leq \left \lvert \Delta \right \rvert + 2\lambda$}, thereby proving the claim by contradiction. Furthermore, since the sum of the negative elements is at least $-\lambda$, it also follows that the number of negative elements in any solution vector~$\boldsymbol{\delta}$ is at most $\lambda$.

Let \mbox{$i_2\in \{0,1,\ldots,\lambda\}$} be the number of negative elements in a given solution vector $\boldsymbol{\delta}$ among the $i_1$ non-zero elements. For given values of $i_1$ and $i_2$, there are $\binom{i_1}{i_2}$ choices for selecting the indices of the negative elements. Let \mbox{$i_3\in \{-\lambda, -\lambda +1, \ldots, 0\}$} be the sum of these \mbox{$i_2$} negative elements. Next, for given values of $i_1$, $i_2$, and $i_3$, we count the number of solutions of the following two equations:
\begin{align}
    \sum_{i : \delta_i<0} \delta_i &= i_3, \label{aeq7}\\
    \sum_{i : \delta_i>0} \delta_i &= \Delta - i_3, \label{aeq8}
\end{align}
where $\left \lvert \{ i : \delta_i<0 \} \right \rvert=i_2$ and $\left \lvert \{ i : \delta_i>0 \}\right \rvert=i_1-i_2$. The number of solutions of \eqref{aeq8} corresponds to the number of integer compositions of $\Delta-i_3$ into $i_1-i_2$ parts, which is given by $\binom{\Delta-i_3-1}{i_1-i_2-1}$~\cite{graham1994concrete}. Similarly, since $i_3$ and all the summands in \eqref{aeq7} are negative, the number of solutions of \eqref{aeq7} is equivalent to the number of integer compositions of $-i_3$ into $i_2$ parts, given by $\binom{-i_3-1}{i_2-1}$.

To finalize the proof, we multiply all possible combinations and sum over all values of $i_1$, $i_2$, and $i_3$ to obtain 
\begin{equation} \label{aeq9}
    \left \lvert \mathcal{P}\right \rvert = \sum_{i_1=0}^{c_1}  \sum_{i_2=0}^{\lambda}  \sum_{i_3=-\lambda}^{0} \binom{N'}{i_1} \binom{i_1}{i_2} \binom{-i_3-1}{i_2-1} \binom{\Delta -i_3-1}{i_1-i_2-1},
\end{equation}
for $\Delta\geq 0$. The expression in Lemma~\ref{prop1} follows immediately from \eqref{aeq9} by applying symmetry with respect to $\Delta$ and a simple variable transformation for $i_3$. 

\vspace{0.2cm}
To upper bound $\left \lvert \mathcal{P}_{\text{general}}\right \rvert$, define the feasible support caps for $i_1$ and $i_2$:
\begin{align*}
i_1^{*}&\triangleq \min\{c_1,\ \max\{\lambda+1,\ |\Delta|+2\lambda\}\}, \\
i_2^{*}&\triangleq  \min\{\lambda,\,i_1\}.
\end{align*}
We have
\begin{align}
|\mathcal {P}|
&=\hspace{-0.1cm} \sum_{i_1=0}^{c_1} \sum_{i_2=0}^{\lambda} \sum_{i_3=0}^{\lambda}
\hspace{-0.1cm} \binom{N'}{i_1}\binom{i_1}{i_2}\binom{i_3-1}{i_2-1}
\binom{|\Delta|+i_3-1}{i_1-i_2-1} \nonumber\\
&=\hspace{-0.1cm} \sum_{i_1=1}^{i_1^{*}} \sum_{i_2=0}^{i_2^{*}} \sum_{i_3=0}^{\lambda}
\hspace{-0.1cm} \binom{N'}{i_1}\binom{i_1}{i_2}\binom{i_3-1}{i_2-1}
\binom{|\Delta|+i_3-1}{\,i_1-i_2-1\,} \tag{a}\\
&\hspace{-0.1cm} \le \sum_{i_1=1}^{i_1^{*}} \sum_{i_2=0}^{i_2^{*}} \sum_{i_3=0}^{\lambda}
\hspace{-0.1cm} \binom{N'}{i_1}\binom{i_1}{i_2}\binom{i_3-1}{i_2-1}\,
\hspace{-0.1cm} \underbrace{\binom{|\Delta|+\lambda-1}{\,i_1-i_2-1\,}}_{\le\ 2^{|\Delta|+\lambda-1}}
\tag{b}\\
&\le 2^{\,|\Delta|+\lambda-1}
\sum_{i_1=1}^{i_1^{*}}\binom{N'}{i_1}
\sum_{i_2=0}^{i_2^{*}}\binom{i_1}{i_2}
\underbrace{\sum_{i_3=0}^{\lambda}\binom{i_3-1}{i_2-1}}_{=\ \tbinom{\lambda}{i_2}}
\tag{c}\\
&= 2^{\,|\Delta|+\lambda-1}
\sum_{i_1=1}^{i_1^{*}}\binom{N'}{i_1}
\underbrace{\sum_{i_2=0}^{\min\{\lambda,i_1\}}\binom{i_1}{i_2}\binom{\lambda}{i_2}}_{=\ \tbinom{i_1+\lambda}{i_1}}
\tag{d}\\
&\le 2^{\,|\Delta|+\lambda-1}
\sum_{i_1=1}^{i_1^{*}}
\underbrace{\binom{N'}{i_1}}_{\le \tfrac{(N')^{i_1}}{i_1!}}\;
\underbrace{\binom{i_1+\lambda}{i_1}}_{\le \tfrac{(\lambda+i_1)^{i_1}}{i_1!}}
\tag{e}\vspace{0.5cm} \\
&\le 2^{\,|\Delta|+\lambda-1}\;
i_1^{*}\;
\frac{(N')^{i_1^{*}}\,(\lambda+i_1^{*})^{i_1^{*}}}{(i_1^{*}!)^{2}}.
\tag{f}
\end{align}

\newpage
\noindent \textbf{Justification of the labeled steps}:
\begin{itemize}
\item[(a)] (\emph{Gate / feasibility}): For a nonzero summand there must exist $i_2\leq i_1$, and \(i_3\!\le\!\lambda\), with \(i_2\!\le\!i_3\) and \(i_1-i_2\!\le\!|\Delta|+i_3\), hence \(i_1\!\le\!\max\{\lambda+1,|\Delta|+2\lambda\}\). Terms with \(i_1>i_1^{*}\) vanish, so the outer sum may be truncated exactly. We then drop the nonnegative \(i_1=0\) term (all-zero pattern), yielding “\(\le\)” and starting at \(i_1=1\).
\item[(b)] (\emph{Replace \(i_3\) by \(\lambda\)}): The map \(i_3\mapsto \binom{|\Delta|+i_3-1}{\,i_1-i_2-1\,}\) is nondecreasing, so its maximum over \(i_3\in[0,\lambda]\) is attained at \(i_3=\lambda\).
\item[(c)] (\emph{Apply \( \binom{a}{b}\le 2^{a}\)}): Apply the standard upper bound \(\binom{|\Delta|+\lambda-1}{\cdot}\le 2^{|\Delta|+\lambda-1}\).
\item[(d)] (\emph{Hockey-stick identity on \(i_3\)}): \(\sum_{i_3=0}^{\lambda}\binom{i_3-1}{i_2-1}=\binom{\lambda}{i_2}\), with \(\binom{-1}{-1}=1\) when \(i_2=0\).
\item[(e)] (\emph{Vandermonde identity on \(i_2\)}): \(\sum_{i_2}\binom{i_1}{i_2}\binom{\lambda}{i_2}=\binom{i_1+\lambda}{i_1}\).
\item[(f)] (\emph{Binomial bounds and collapse sum}): Use \(\binom{n}{k}\le n^{k}/k!\) on both \(\binom{N'}{i_1}\) and \(\binom{i_1+\lambda}{i_1}\). The remaining finite sum over \(i_1\in\{1,\dots,i_1^{*}\}\) is upper-bounded by \(i_1^{*}\) times its largest-index term. Note that $N'=K+c_1$ and $i_1^{*}\leq c_1$, so \(i_1^{*}\) always maximizes $\binom{N'}{i_1}$.
\end{itemize}

\section{Proof of Claim~\ref{jointE2}} \label{appB}
Fix integers $z\in[0,\ell N']$ and $\Delta\in[-\ell N',\ell N']$. Consider the following two cases.
\paragraph{Case 1, \(\Delta\ge 0\)}
Here the cumulative offset is non–negative, so
\[
  \tilde{\rv D}^{+}= \Delta+z,\qquad
  \tilde{\rv D}^{-}=z,\qquad
  \tilde{\rv Z}=z .
\]

Choose exactly \(j\) indices (out of \(N'\)) whose offsets are strictly negative.  
The number of such choices is \(\binom{N'}{j}\).  
For the chosen indices the magnitudes must sum to \(z\); the probability of that event is the \(z\)-th entry of the \(j\)-fold convolution \(\bar{\q}^{(*j)}\).  
For the remaining \(N'-j\) indices the offsets are non–negative and must sum up to exactly \(\Delta+z\); the probability of that is the \((\Delta+z)\)-th entry of \(\p^{(*\,N'-j)}\).  
Multiplying and summing over all values of \(j\) gives the asserted formula. The feasible range of \(j\) is
\begin{itemize}
  \item \emph{Lower bound:} each negative offset contributes at most \(\ell\), hence \(j\ge\lceil z/\ell\rceil\);
  \item \emph{Upper bound 1:} at least one unit of magnitude is contributed by each of the \(j\) negative offsets, so \(j\le z\);
  \item \emph{Upper bound 2:} the remaining \(N'-j\) offsets must be able to generate total magnitude \(\Delta+z\le\ell(N'-j)\), giving
        \(j\le N'-\lceil(\Delta+z)/\ell\rceil\).
\end{itemize}
Therefore, $j_{\min}=\lceil z/\ell\rceil$ and taking the minimum of the two upper bounds yields $j_{\max}$.

\paragraph{Case 2, \(\Delta<0\)}
Now \(\tilde{\rv D}^{-}= |\Delta|+z\), \(\tilde{\rv D}^{+}=z\).  
The same argument with the roles of \(\p\) and \(\q\) interchanged, and the same feasible range of \(j\), produces the second part of the claim.

Since the two cases are exhaustive and disjoint, the expression in the claim gives the exact joint PMF of \((\tilde{\rv Z},\tilde{\rv D})\).\hfill$\square$

\section{Upper Bound on the Probability of Event $\mathcal{E}_3$ for a Repetition Code with Sliding-Window Decoding} \label{appC}

\subsection*{1. Preliminaries}
Consider a $t$-fold repetition code that encodes a message $\tilde{\u} \in \mathbb{F}_2^{\tilde{k}}$ into a codeword $\tilde{\x} \in \mathbb{F}_2^{\tilde{n}}$, where each bit $\tilde{u}_b$, $b \in [\tilde{k}]$, is repeated $t$ times. The resulting codeword has length $\tilde{n} = t \tilde{k}$ and consists of $\tilde{k}$ {\em blocks} of size $t$:
\begin{equation} \label{eqApp1}
  \tilde{\x}=(\tilde{x}_1,\ldots,\tilde{x}_{\tilde{n}})=(\underbrace{\overbracket{\tilde{u}_1,\ldots,\tilde{u}_1}^{\text{block 1}}}_{t\text{ times}}\, , \ldots \ldots \, ,
  \underbrace{\overbracket{\tilde{u}_{\tilde{k}},\ldots,\tilde{u}_{\tilde{k}}}^{\text{block }\tilde{k}}}_{t\text{ times}}).
\end{equation}

In the GC+ code construction, the message $\tilde{\u}$ corresponds to the $\tilde{k} = c_2 \ell$ check parity bits, and the associated repetition codeword $\tilde{\x}$ is appended as a suffix of the GC+ codeword. Let $\x$ be an input GC+ codeword to the edit channel described in Section~\ref{model}, and let $\y$ be the corresponding output. We consider a sliding-window decoding approach in which, to recover the check parities, the decoder scans the last $\tilde{n}$ bits of $\y$ and applies majority voting over $\tilde{k}$ non-overlapping windows, each of length $t$. Due to potential deletions and insertions, these windows may be misaligned with the true block boundaries. Consequently, each window may include extraneous bits from neighboring repetition blocks or from outside the suffix region (i.e., bits not belonging to the repetition code).

To analyze the performance of sliding-window decoding of the repetition suffix, we consider a simplified yet equivalent scenario. Specifically, we assume that a repetition codeword $\tilde{\x}$ of the form given in~\eqref{eqApp1} is sent in isolation through the same edit channel, followed by additional random bits appended to its right, resulting in an output $\tilde{\y}$. These appended bits model the possibility that the decoder reads beyond the actual suffix boundary. We further assume the bits of $\tilde{\u}$ and the appended bits are i.i.d. Bernoulli$(1/2)$.\footnote{In Appendix C, we treat $\tilde{\u}$, $\tilde{\x}$, and $\tilde{\y}$ as random vectors, though we do not use the sans-serif notation for them to avoid notational clutter.}

To estimate $\tilde{u}_b$, the decoder applies a majority vote over the bits of $\tilde{\y}$ whose positions lie within the window
\begin{equation}
  \mathcal{W}^{(b)} \;\triangleq \; [(b-1)t+1,bt], \quad b\in [\tilde{k}].
\end{equation}
Note that these decoding windows remain fixed for any $\tilde{\y}$, which can lead to decoding errors due to misalignment. Define the random variable
\begin{equation}
\rv{S}^{(b)} \triangleq \left \lvert \{ i \in \mathcal{W}^{(b)} : \tilde{y}_i = \tilde{u}_b \} \right \rvert \in [0,t],
\end{equation}
as the number of {\em correct} votes for block $b\in [\tilde{k}]$, i.e., the number of bits of $\tilde{\y}$ in the decoding window $\mathcal{W}^{(b)}$ that match the correct value $\tilde{u}_b$. A decoding error occurs if \mbox{$\rv{S}^{(b)} < \lceil (t+1)/2 \rceil$} for any $b \in [\tilde{k}]$. To upper-bound $\Pr(\mathcal{E}_3)$, we apply the union bound:
\begin{equation} \label{eqU}
  \Pr \Big(\bigcup_{b=1}^{\tilde{k}} \big\{ \rv{S}^{(b)} < \left \lceil \tfrac{t+1}{2} \right \rceil \big\} \Bigr)
  \le
  \sum_{b=1}^{\tilde{k}} \Pr \Big(\rv{S}^{(b)} < \left \lceil \tfrac{t+1}{2} \right \rceil \Big).
\end{equation}

\subsection*{2. Restricting Analysis to Three Blocks $\{b-1,b,b+1\}$}

To simplify the analysis of $\rv{S}^{(b)}$ for a given $b\in [\tilde{k}]$, we restrict attention to the bits of $\tilde{\x}$ belonging to block $b$ itself, and its immediate neighbors, blocks $b-1$ and $b+1$. While bits from more distant blocks (e.g., block~$b-2$ or $b+2$) could theoretically fall into $\mathcal{W}^{(b)}$ in $\tilde{\y}$ due to large offsets and contribute to $\rv{S}^{(b)}$, the probability of such occurrences is typically small. Furthermore, since $\rv{S}^{(b)}$ by definition counts the number of {\em correct} votes, ignoring the contribution of distant blocks can only decrease $\rv{S}^{(b)}$, thereby leading to an upper bound on the decoding error probability. Note that:
\begin{itemize}
\item \emph{If $b=1$}, block $(b-1)$ does not exist, so the analysis includes only blocks $1$ and $2$.
\item \emph{If $b = \tilde{k}$}, block $(\tilde{k} + 1)$ lies beyond the repetition code.  In this case, block $(\tilde{k}+1)$ is treated as consisting of $t$ i.i.d. bits, each matching $\tilde{u}_{\tilde{k}}$ with probability $0.5$.  
\end{itemize}
\subsection*{3. Law of Total Probability Over Four Cases} \label{polarity}
Since the message bits $\tilde{u}_{b-1}$, $\tilde{u}_b$, and $\tilde{u}_{b+1}$ are assumed to be i.i.d. Bernoulli(0.5), we have
\begin{equation}
  \Pr(\tilde{u}_{b-1} = \tilde{u}_b) = \Pr(\tilde{u}_{b+1} = \tilde{u}_b) = 0.5.
\end{equation}
We thus break down the analysis into four equiprobable cases:
\begin{enumerate}
\item {\bf Case A}: $\tilde{u}_{b-1}=\tilde{u}_b$, $\tilde{u}_{b+1}=\tilde{u}_b$,
\item {\bf Case B}: $\tilde{u}_{b-1}\neq \tilde{u}_b$, $\tilde{u}_{b+1}=\tilde{u}_b$,
\item {\bf Case C}: $\tilde{u}_{b-1}=\tilde{u}_b$, $\tilde{u}_{b+1}\neq \tilde{u}_b$,
\item {\bf Case D}: $\tilde{u}_{b-1}\neq \tilde{u}_b$, $\tilde{u}_{b+1}\neq \tilde{u}_b$.
\end{enumerate}
Within each case, all bits from block $b-1$ are either identical to $\tilde{u}_b$ (if $\tilde{u}_{b-1} = \tilde{u}_b$) or flipped relative to it (if $\tilde{u}_{b-1} \neq \tilde{u}_b$), and similarly for block $b+1$. These cases allow us to condition on how a neighboring block contributes to $\rv{S}^{(b)}$ if any of its bits fall into $\mathcal{W}^{(b)}$ in $\tilde{\y}$. Each case is weighted by $1/4$ in the final decoding error probability via the law of total probability.

\subsection*{4. Recurrence Relation}
We evaluate the PMF of $\rv{S}^{(b)}$ by formulating a recurrence relation that tracks the evolution of the correct vote count as the bits from blocks $\{b-1, b, b+1\}$ in $\tilde{\x}$ are sequentially processed through the edit channel. 
\paragraph{Definitions}
For a block $b\in[\tilde{k}]$, define
\begin{equation}
m_b \;\triangleq\;
\begin{cases}
(b-2)t,& b\geq 2,\\
0,& b<2,
\end{cases}\;;
\quad
l_b \;\triangleq\;
\begin{cases}
2t,& b=1,\\
3t,& b\ge 2,
\end{cases}\;,
\end{equation}
where $m_b$ is the number of bits in $\tilde{\x}$ that precede the start of block $b-1$, and $l_b$ is the total number of bits in blocks~\mbox{$\{b-1, b, b+1\}$}. We focus on the contribution of bits in the segment $\tilde{\x}_{[m_b+1,m_b+l_b]}$ to the correct vote count~$\rv{S}^{(b)}$. The analysis proceeds bit by bit, with the {\em local} position $i\in [l_b]$ indexing each bit, and $\rv{S}_i^{(b)}$ denoting the total vote count accumulated after processing all bits up to $\tilde{x}_{m_b+i}$ through the channel.

The position of the output corresponding to bit~$\tilde{x}_{m_b+i}$ in~$\tilde{\y}$ depends on the cumulative offset induced by the channel up to that bit. We track this offset through the random variable~\mbox{$\tilde{\rv{D}}_j \in [-j,j]$}, defined as the net number of insertions minus deletions introduced by the channel on the first $j$ bits of $\tilde{\x}$, i.e., $\tilde{\x}_{[1,j]}$. Accordingly, the {\em global} position of the output corresponding to~$\tilde{x}_{m_b+i}$ in~$\tilde{\y}$ is $m_b + i + \tilde{\rv{D}}_{m_b+i-1}$.

For $i \in [l_b]$, $\Delta\in [-(m_b+i),m_b+i]$, and $s\in[0,t]$, define the state function
\begin{equation}
h^{(b)}_{i}(s,\Delta) \triangleq \Pr(\rv{S}_i^{(b)}=s,\tilde{\rv{D}}_{m_b+i}= \Delta),
\end{equation}
which gives the joint probability that, after processing all bits up to $\tilde{x}_{m_b+i}$: {\em (i)}~exactly $s$ correct votes have been accumulated for block~$b$; and {\em (ii)}~the cumulative offset is $\Delta$. The initial condition at $i = 0$ is
\begin{equation} \label{init}
h^{(b)}_0(s,\Delta) = \begin{cases}
\Pr(\tilde{\rv{D}}_{m_b}=\Delta),& s=0,\\
0,& s>0,
\end{cases}\;,
\end{equation}
where the PMF of $\tilde{\rv{D}}_{m_b}\in [-m_b,m_b]$ is given by~\eqref{eqDd} with $\ell$ replaced by $m_b$.

\paragraph{State Transitions}
Suppose that at step $i-1$, the state is $(s_{i-1}, \Delta_{i-1})$, with probability $h_{i-1}^{(b)}(s_{i-1}, \Delta_{i-1})$. The next input bit to the channel at step $i$ is $\tilde{x}_{m_b+i}$, and the corresponding output is $\tilde{y}_{i'}$, where \mbox{$i'\triangleq m_b + i + \Delta_{i-1}$} is the {\em global} position of the output in $\tilde{\y}$. The transition from $(s_{i-1},\Delta_{i-1})$ to $(s_i,\Delta_i)$ depends on the following factors:
\begin{itemize}
\item {\bf Output position}: whether $\tilde{y}_{i'}$ falls into the decoding window~$\mathcal{W}^{(b)}$ of block $b$, defined as  
  \[
  \omega_i \triangleq \mathds{1}_{\{i' \in \mathcal{W}^{(b)}\}},
  \]
  which is deterministic given $i$ and $\Delta_{i-1}$.
\item {\bf Input bit match}: whether the input bit~$\tilde{x}_{m_b+i}$ matches~$\tilde{u}_b$, defined as  
  \[
  \beta_i^{(c)} \triangleq \mathds{1}_{\{\tilde{x}_{m_b+i} = \tilde{u}_b\}},
  \]
  where \(c \in \{\text{A},\text{B},\text{C},\text{D}\}\) denotes the polarity case for bits originating from blocks~$b-1$ and~$b+1$ as described in Section 3 of this appendix. This indicator is deterministic given $i$ and $c$ for all $b\in[\tilde{k}]$, except in the special case $b=\tilde{k}$ (last block), where $\beta_i^{(c)}\sim\text{Bernoulli}(0.5)$ for \mbox{$i\in[2t+1,3t]$}, reflecting the assumption (Section 2) that block $(\tilde{k}+1)$ consists of $t$ uniform i.i.d. bits.
\item {\bf Channel event}: the error event affecting~$\tilde{x}_{m_b+i}$, represented by \mbox{$e\in \{\text{no error}, \text{deletion}, \text{insertion}, \text{substitution}\}$}, which determines the updated offset $\Delta_i$ and the potential impact of $\tilde{y}_{i'}$ on $s_i$.
\end{itemize}
For a given pair $(c,e)$, representing the polarity case and error event, we describe the state update at step $i$ through a mapping
\[
   \tau^{(c,e)}:(i,s_{i-1},\Delta_{i-1})\longmapsto(s_i,\Delta_i),
\]
which depends on the three aforementioned factors. Next, we characterize this mapping for each channel event $e$ in terms of $\omega_i$ and $\beta_i^{(c)}$. For brevity, we write $(s,\Delta)$ and $(s',\Delta')$ as shorthand for $(s_{i-1},\Delta_{i-1})$ and $(s_i,\Delta_i)$, respectively. 
\begin{enumerate}

\item \textbf{No error} (with probability $1-P_d-P_i-P_s$): The channel outputs one bit identical to the input, i.e., \mbox{$\tilde{y}_{i'}=\tilde{x}_{m_b+i}$}. The cumulative offset remains unchanged ($\Delta'=\Delta$). The updated state $(s',\Delta')$ is given by the mapping
\begin{align*}
\tau^{(c,\;\text{no error})}(i,s,\Delta) =
\begin{cases}
\vcenter{\hbox{$(s+1, \Delta)$}}, &
\parbox[t]{2cm}{
if \( \omega_i = 1 \), \\ and \( \beta_i^{(c)} = 1 \),} \vspace{0.1cm} \\
(s, \Delta), & \text{otherwise}.
\end{cases}
\end{align*}

\item \textbf{Deletion} (with probability $P_d$): The channel produces no output. The cumulative offset is decreased by one ($\Delta'=\Delta-1$) and the vote count is unchanged ($s'=s$). The updated state $(s',\Delta')$ is thus given by the mapping
\begin{align*}
\tau^{(c,\;\text{deletion})}(i,s,\Delta)=(s,\Delta-1).
\end{align*}

\item \textbf{Substitution} (with probability $P_s$): The channel outputs one bit, which is the complement of the input bit, i.e., \mbox{$\tilde{y}_{i'}=1-\tilde{x}_{m_b+i}$}. The cumulative offset remains unchanged ($\Delta'=\Delta$). The updated state $(s',\Delta')$ is given by the mapping
\begin{align*}
\tau^{(c,\;\text{substitution})}(i,s,\Delta) =
\begin{cases}
\vcenter{\hbox{$(s+1, \Delta)$}}, &
\parbox[t]{2cm}{
if \( \omega_i = 1 \), \\ and \( \beta_i^{(c)} = 0 \),} \vspace{0.1cm} \\
(s, \Delta), & \text{otherwise}.
\end{cases}
\end{align*}

\item \textbf{Insertion} (with probability $P_i$): The channel outputs two bits $(\tilde{y}_{i'},\tilde{y}_{i'+1})= (\sigma, \tilde{x}_{m_b+i} )$, where the first bit \mbox{$\sigma\sim\text{Bernoulli}(0.5)$} is the inserted one, and the second bit is identical to the input. The cumulative offset is increased by one ($\Delta'=\Delta+1$). Define the following mutually exclusive events
\begin{align*}
\Psi_i^1 &\triangleq \{(\omega_i, \omega_{i+1}) = (1,1)\, \land\, \beta_i^{(c)}=1\, \land\, \sigma=\tilde{x}_{m_b+i}\}, \\
\Psi_i^2 &\triangleq \{(\omega_i, \omega_{i+1}) = (1,1)\, \land\, \beta_i^{(c)}=0\, \land\, \sigma\neq \tilde{x}_{m_b+i}\}, \\
\Psi_i^3 &\triangleq \{(\omega_i, \omega_{i+1}) = (1,1)\, \land\, \beta_i^{(c)}=1\, \land\, \sigma\neq \tilde{x}_{m_b+i}\}, \\
\Psi_i^4 &\triangleq \{(\omega_i, \omega_{i+1}) = (1,0)\, \land\, \beta_i^{(c)}=1\, \land\, \sigma= \tilde{x}_{m_b+i}\}, \\
\Psi_i^5 &\triangleq \{(\omega_i, \omega_{i+1}) = (1,0)\, \land\, \beta_i^{(c)}=0\, \land\, \sigma\neq \tilde{x}_{m_b+i}\}, \\
\Psi_i^6 &\triangleq \{(\omega_i, \omega_{i+1}) = (0, 1)\, \land\, \beta_i^{(c)}=1\}, \\
\Psi_i^{[2,6]} &\triangleq \Psi_i^2\, \lor\, \Psi_i^3\, \lor\, \Psi_i^4\, \lor\, \Psi_i^5 \lor\, \Psi_i^6.
\end{align*}
\noindent The updated state $(s',\Delta')$ is given by the mapping
\[
\tau^{(c,\;\text{insertion})}(i,s,\Delta) =
\begin{cases}
(s+2, \Delta+1), &
\text{if } \Psi_i^1 = 1, \\[4pt]
(s+1, \Delta+1), &
\text{if } \Psi_i^{[2,6]} = 1, \\[4pt]
(s, \Delta+1), &
\text{otherwise.}
\end{cases}
\]
\end{enumerate}

\paragraph{Markov recursion}
For every $i\in[l_b]$ and pair of states $(s,\Delta)$ and $(s',\Delta')$, the transition probability is given by
\begin{equation*} \label{eq:transition-prob}
\pi_i\bigl( (s',\Delta') \mid (s,\Delta)\bigr) \triangleq
\sum_{c,e} 
\Pr(c,e) \;  \mathbb{E}\!\left[ \mathds{1}_{\{(s',\Delta')=\tau_i^{(c,e)}(s,\Delta)\}} \right],
\end{equation*}
where the expectation is over the auxiliary random variables $\sigma \sim\text{Bernoulli}(0.5)$ when the channel event $e$ is an insertion,
and $\beta_i^{(c)}\sim\text{Bernoulli}(0.5)$ only for
\(b=\tilde{k}\) and \(i\in[2t{+}1,3t]\);
otherwise the indicator is deterministic.

Since the four polarity cases are equiprobable and independent of the channel events, we have 
\begin{equation*} \label{eq:transition-prob}
\pi_i\bigl( (s',\Delta') \mid (s,\Delta)\bigr) =
\frac{1}{4} \sum_{c,e} 
 \Pr(e) \; \mathbb{E}\!\left[ \mathds{1}_{\{(s',\Delta')=\tau_i^{(c,e)}(s,\Delta)\}} \right],
\end{equation*}
where \mbox{$\Pr(\text{no error})=1-P_d-P_i-P_s$}, \mbox{$\Pr(\text{deletion})=P_d$}, \mbox{$\Pr(\text{insertion})=P_i$}, and \mbox{$\Pr(\text{substitution})=P_s$}. 

Starting from the initial distribution given in~\eqref{init},
the probability mass function evolves according to the Markov recursion
\begin{equation}\label{eq:forward}
  h^{(b)}_{i}(s',\Delta')
  \;=\;
  \sum_{(s,\Delta)}
  h^{(b)}_{i-1}(s,\Delta)\;
  \pi_i\bigl( (s',\Delta') \mid (s,\Delta)\bigr),
\end{equation}
for $i\in [l_b]$.
\subsection*{5. Marginalization, Error Probability, and Union Bound}
After processing all $l_b$ bits from blocks $\{b-1,b,b+1\}$, we marginalize over all offsets $\Delta\in [-(m_b+l_b),m_b+l_b]$ to obtain the PMF of $\rv{S}^{(b)}$:
\begin{equation}
   \Pr(\rv{S}^{(b)}=s) = \sum_{\Delta} h_{l_b}^{(b)}(s,\Delta).
\end{equation}
These steps are repeated for each block index $b \in [\tilde{k}]$ to determine the full set of PMFs required to apply the union bound in~\eqref{eqU} and evaluate the upper bound:
\begin{align}
 \Pr(\mathcal{E}_3)\leq \sum_{b=1}^{\tilde{k}} \Pr \Big(\rv{S}^{(b)} < \left \lceil \tfrac{t+1}{2} \right \rceil \Big).
\end{align}

\end{document}